\documentclass[11pt,reqno]{amsart}
\usepackage{amsaddr}
\usepackage{setspace}
\usepackage{graphicx}
\usepackage{dsfont}
\usepackage{enumerate}
\usepackage{mathtools}
\usepackage{amsmath}
\usepackage{amssymb}
\usepackage{nicefrac}
\usepackage{tikz}
\usepackage{hyperref}
\usepackage[breakable]{tcolorbox}
\usepackage{algorithm}
\usepackage{algpseudocode}
\usepackage{xcolor}

\usepackage[backend=biber, style=alphabetic, backref=true, hyperref=true]{biblatex}
\usepackage[babel,english=british]{csquotes}
\DefineBibliographyStrings{english}{%
    backrefpage  = {cited on p.}, 
    backrefpages = {cited on pp.} 
}
\newcommand{\norm}[1]{\left\lVert#1\right\rVert}
\newcommand{\IR}{\mathbb{R}}
\newcommand{\IN}{\mathbb{N}}

\newcommand{\IC}{\mathbb{C}}

\newcommand{\VC}{\operatorname{VCdim}}

\usepackage{amsthm}
\newtheorem{theorem}{Theorem}[section]

\newtheorem{corollary}[theorem]{Corollary}

\newtheorem{lemma}[theorem]{Lemma}
\newtheorem{definition}[theorem]{Definition}

\theoremstyle{definition}
\newtheorem{example}[theorem]{Example}
\newtheorem{remark}[theorem]{Remark}

\numberwithin{equation}{section}

\DeclareMathOperator{\tr}{tr}

\renewenvironment{proof}{{\bfseries Proof:}}{\hfill $\square$}

\begin{document}

\title{Binary Classification with Classical Instances and Quantum Labels}
\author{Matthias C. Caro}   
\address{Technical University of Munich, Germany, Department of Mathematics\\
Munich Center for Quantum Science and Technology (MCQST), Munich, Germany}
\email{caro@ma.tum.de, \emph{ORCiD:} \href{https://orcid.org/0000-0001-9009-2372}{0000-0001-9009-2372}}
\date{\today}

\begin{abstract}
In classical statistical learning theory, one of the most well studied problems is that of binary classification. The information-theoretic sample complexity of this task is tightly characterized by the Vapnik-Chervonenkis (VC) dimension. A quantum analog of this task, with training data given as a quantum state has also been intensely studied and is now known to have the same sample complexity as its classical counterpart.\\
We propose a novel quantum version of the classical binary classification task by considering maps with classical input and quantum output and corresponding classical-quantum training data. We discuss learning strategies for the agnostic and for the realizable case and study their performance to obtain sample complexity upper bounds. Moreover, we provide sample complexity lower bounds which show that our upper bounds are essentially tight for pure output states. In particular, we see that the sample complexity is the same as in the classical binary classification task w.r.t.~its dependence on accuracy, confidence and the VC-dimension.
\end{abstract}

\maketitle


\section{Introduction}
The fields of machine learning and of quantum computation provide new ways of looking at computational problems and have seen a significant increase in academic as well as practical interest since their origins in the $1970$s and $1980$s. More recently, attention was directed to paths for combining ideas from these two fruitful research areas. This gave rise to new approaches under different names such as ``quantum machine learning'' or ``quantum learning theory''.\\

In classical statistical learning theory, one of the most influential frameworks is that of probably approximately correct (PAC) learning due to \cite{Vapnik.1971} and \cite{Valiant.1984}. It is particularly well studied for the task of binary classification. For this problem the so-called VC-dimension \cite{Vapnik.1971} is known to characterize the sample complexity of learning a function class \cite{Blumer.1989,Hanneke.2016}. Motivated by these strong theoretical results, a quantum analog of this problem was soon defined and studied in a series of papers (an overview over which is given in \cite{Arunachalam.2017}), which culminated in \cite{Arunachalam.2018}. Therein it is shown that the information-theoretic complexity of the task of quantum PAC learning a $0$-$1$-valued function class is characterized by \enlargethispage{\baselineskip} the VC-dimension in exactly the same way as for the classical scenario.\\

The scenario studied in \cite{Arunachalam.2018} assumes the training data available to the learner to be given in a specific quantum form and allows the learner to perform quantum computational operations on that training data. The functions to be learned, however, still map classical inputs to classical outputs. We propose a different quantum version of the binary classification task by not only considering the possibility of quantum training data but by allowing the objects to be learned to be inherently quantum. More specifically, we consider functions that map classical inputs to one of two possible quantum output states (``quantum labels''). These maps describe state preparation procedures. A more general learning task of this type, for which our problem can be seen as a toy model, could be relevant for cases in which state preparation is either costly or time-consuming, e.g., preparing thermal states at low temperatures (see \cite{BrandaoF.G.S.L..2019, Chowdhury.2020} and references therein). Here, one could first produce sample data, learn a predictor, and then reproduce the preparation more efficiently using the predictor.

\subsection{Main Results}
We consider maps $f:\mathcal{X}\to\{ \sigma_0,\sigma_1\}$ that assign to points in a classical input space $\mathcal{X}$ one of two labelling quantum states $\{\sigma_0,\sigma_1\}$. (Here, $\sigma_0$ and $\sigma_1$ are, in general, mixed states described by density matrices.) Let $\mathcal{F}$ be a function class consisting of such functions.  We assume the training data to be given as a classical-quantum state about which, according to the laws of quantum theory, we can only gain information by performing measurements.\\

Our learning model is that of PAC-learning with accuracy $\varepsilon$ and confidence $\delta$. Here, we require a learning algorithm, given as input classical-quantum training data generated according to some unknown underlying distribution, to output with probability $\geq 1-\delta$ over the choice of training data a hypothesis that achieves accuracy $\varepsilon$. (Accuracy is measured in terms of the trace distance.)\\
We present a learning strategy that $(\varepsilon,\delta)$-PAC learns $\mathcal{F}\subseteq\{ f:\mathcal{X}\to\{ \sigma_0,\sigma_1\}\}$ in the agnostic scenario from classical-quantum training data of size $\mathcal{O}\left( \tfrac{d}{\varepsilon^2} + \tfrac{\log\nicefrac{1}{\delta}}{\varepsilon^2}\right)$, where $d$ is the VC-dimension of the $\{0,1\}$-valued function class $\tilde{\mathcal{F}}\subseteq\{\tilde{f}:\mathcal{X}\to \{0,1\}\}$ induced by $\mathcal{F}$ via $\sigma_i\mapsto i$, $i=0,1$. Here, ``agnostic'' means that there need not be a function in $\mathcal{F}$ that would achieve perfect accuracy. We also show that solving this learning problem requires training data size $\Omega\left( \tfrac{d}{\varepsilon^2} + \tfrac{\log\nicefrac{1}{\delta}}{\varepsilon^2} \right)$, so our strategy is optimal w.r.t.~the sample complexity dependence on $\varepsilon$, $\delta$ and $d$.\\

For the realizable scenario of the quantum learning problem, i.e., under the assumption that perfect accuracy can be achieved using $\mathcal{F}$, we prove a sample complexity upper bound of $$\mathcal{O}\left(\tfrac{1}{\varepsilon (1-2\max\lbrace \tr[E_0\sigma_1],\tr[E_1\sigma_0]\rbrace)^2} \left( d + \log\nicefrac{1}{\delta}\right)\right),$$ where $\{ E_0,E_1\}$ is the Holevo-Helstrom measurement for distinguishing $\sigma_0$ and $\sigma_1$, and a sample complexity lower bound of $\Omega\left(\tfrac{d}{\varepsilon} + \tfrac{\log\nicefrac{1}{\delta}}{\varepsilon}\right)$. Also here, these bounds coincide w.r.t.~their dependence on $\varepsilon$, $\delta$ and $d$. The prefactor $(1-2\max\lbrace \tr[E_0\sigma_1],\tr[E_1\sigma_0]\rbrace)^{-2}$ in the upper bound comes from our procedure trying to distinguish $\sigma_0$ and $\sigma_1$ by measuring single copies. (Note: Even though we formulate this in terms of the Holevo-Helstrom measurement, we could use any other two-outcome POVM $\{ \tilde{E}_0,\tilde{E}_1\}$ that satisfies $\max\lbrace \tr[\tilde{E}_0\sigma_1],\tr[\tilde{E}_1\sigma_0]\rbrace<\nicefrac{1}{2}.$)\\

In proving the sample complexity upper bound for the realizable scenario, we combine algorithms from \cite{Laird.1988} and \cite{Hanneke.2016} to show that $\mathcal{O}\left(\frac{1}{\varepsilon (1-2\eta_b)^2} \left( d + \log\nicefrac{1}{\delta}\right)\right)$ classical examples with two-sided classification noise, i.e., in which each label is flipped with probability given by a noise rate, suffice for classical $(\varepsilon,\delta)$-PAC learning a function class of VC-dimension $d$ in the realizable scenario if the noise rate is bounded by $\eta_b<\nicefrac{1}{2}$. This upper bound has, to the best of our knowledge, not been observed before and, when combined with the lower bound from \cite{Arunachalam.2018}, establishes the optimal sample complexity of this classical noisy learning problem.\\

As is common in statistical learning theory, our main focus lies on the information-theoretic complexity of the learning problem, i.e., the necessary and sufficient number of quantum examples, whereas we do not discuss the computational complexity. Our proposed strategies are ``semi-classical'' in the following sense: After initially performing tensor-product measurements, in which each tensor factor is a two-outcome POVM, the remaining computation is done by a classical learning algorithm. In particular, the procedure does not require (possibly hard to implement) joint measurements and its computational complexity will be determined by the (classical) computational complexity of the classical learner used as a subroutine.

\subsection{Overview over the Proof Strategy}
We first sketch how we obtain the sample complexity upper bounds. We propose a simple (semi-classical) procedure that consists of first performing local measurements on the quantum part of the training data examples to obtain classical training data and then applying a classical learning algorithm.\\
We observe that the learning problem for which the classical learner is applied, can then be viewed as a classical binary classification problem with two-sided classification noise, i.e., in which the labels are flipped with certain error probabilities determined by the outcome probabilities of the performed quantum measurements. Therefore, we have reduced our problem to obtaining sample complexity upper bounds for a classical learning problem with noise.\\
In the general (so-called \emph{agnostic}) case, we can use known sample complexity bounds formulated in terms a complexity measure called \emph{Rademacher complexity} to show that classical empirical risk minimization w.r.t.~a suitably modified loss function (as suggested in \cite{Natarajan.2013}) achieves optimal sample complexity for this classical learning problem with noise.\\
In the \emph{realizable} case, i.e., under the assumption that any non-noisy training data set can be perfectly represented by some hypothesis in our class $\tilde{\mathcal{F}}$, the optimal sample complexity for binary classification with two-sided classification noise has not been established in the literature. We combine ideas from \cite{Laird.1988} and \cite{Hanneke.2016} to exhibit an algorithm that achieves information-theoretic optimality for this scenario.\\

To obtain the sample complexity lower bounds, we apply ideas from \cite{Arunachalam.2018}. Namely, we observe that for sufficiently small accuracy parameter, any quantum strategy that solves our learning problem indeed has to be able to distinguish between the possible different training data states with high success probability. \\
In the simple case of distinguishing between two quantum states, arising from two different ``hard-to-distinguish'' underlying distributions, this probability can be upper-bounded in terms of the trace distance of the states. 
In the more general case of many states, we do not study this success probability directly. Instead, we consider the information contained in the quantum training data about the choice of the underlying distribution, again chosen out of a set of ``hard-to-distinguish'' distributions.

\subsection{Related Work}
\cite{Bshouty.1998} introduced a notion of quantum training data for learning problems with classical concepts and used it to learn DNF (Disjunctive Normal Form) formulae w.r.t.~the uniform distribution. This was extended to product distributions by \cite{Kanade.2019}. Using ideas from Fourier-based learning, this type of quantum training data was also studied in the context of fixed-distribution learning of Boolean linear functions \cite{Bernstein.1993,Cross.2015,Riste.2017,Grilo.20180410,Caro.2020}, juntas \cite{Atc.2007}, and Fourier-sparse functions \cite{Arunachalam.2019}. \cite{Arunachalam.2017} and \cite{Arunachalam.20190307} study the limitations of these quantum examples. A broad overview over work on quantum learning classical functions is given in \cite{Arunachalam.2017}.\\

Also for the model of learning from membership queries, a quantum counterpart can be considered. \cite{Servedio.2004} showed that the number of required classical queries is at most polynomially larger than the number of required quantum queries. Recently, this polynomial relation was improved upon in \cite{Arunachalam.2019}. A more specific scenario, namely that of learning multilinear polynomials more efficiently from quantum membership queries, is studied in \cite{Montanaro.2012}.\\
Similarly, also a quantum counterpart of the classical model of statistical query learning can be defined. This was recently studied in \cite{Arunachalam.2020}.\\

Another line of research at the intersection of learning theory and quantum information focuses on applying classical learning to concept classes arising from quantum theory, e.g., from states or measurements. This was initiated by \cite{Aaronson.2007} and studied further by \cite{Cheng.2016, Aaronson.2018}, and \cite{Aaronson.2018b}.\\

Our learning model is similar to the one studied in \cite{Chung.2018}. Also there, the inputs are assumed to be classical and the outputs are quantum states. The crucial difference to our scenario is that we assume that there are only two possible label states and these are known in advance. In \cite{Chung.2018}, there can be a continuum of possible label states.\\
Our additional assumption allows us to study infinite function classes $\mathcal{F}$, whereas the results in \cite{Chung.2018} are for classes of finite size. (We expect that the reasoning of \cite{Chung.2018} can be extended to infinite classes using the so-called ``growth function'' when restricting to a finite set of possible target states. This might lead to a learning procedure that can be applied in our scenario without prior knowledge of the possible quantum label states.) As a further difference between the approaches, whereas the strategy of \cite{Chung.2018} requires the ability to perform measurements in random orthonormal bases, the measurements in our procedures can be taken to be fixed and of product form and are thus potentially easier to implement.\\

The classical problems to which our quantum learning problems are reduced, are problems of learning from noisy training data. These were first proposed by \cite{Angluin.1988, Laird.1988} and studied further, e.g., by \cite{Aslam.1996, CesaBianchi.1999} and \cite{Natarajan.2013}.

\subsection{Structure of the Paper}
In Section $2$ we recall some notions from learning theory as well as from quantum information and computation. The central learning problem of this contribution is formulated in Section $3$. The next Section exhibits strategies for solving the task and establishes sample complexity upper bounds. In doing so, we derive a tight upper bound on the sample complexity of classical binary classification with two-sided classification noise (see Appendix \ref{SctNoisyBinaryClassification}). The quantum sample complexity upper bounds are complemented by lower bounds in Section $5$. We conclude with open questions and the references.

\section{Preliminaries}\label{SctPrelims}

\subsection{Basics of Quantum Information and Computation}\label{SbSctBasicQIT}
A finite-dimensional quantum system is described by a \emph{(mixed)} state and mathematically represented by a \emph{density matrix} of some dimension $d\in\IN$, i.e., an element of $\mathcal{S}(\IC^d):=\lbrace \rho\in\IC^{d\times d}\ |\ \rho\geq 0, \tr[\rho]=1\rbrace$. Here, $\rho\geq 0$ means that $\rho$ is a self-adjoint and positive semidefinite matrix. The extreme points of the convex set $\mathcal{S}(\IC^d)$ are the rank-$1$ projections, the \emph{pure states}. We employ Dirac notation to denote a unit vector $\psi\in\IC^d$ also by $|\psi\rangle\in\IC^d$ and the corresponding pure state by $|\psi\rangle\langle\psi|$.\\

To make an observation about a quantum system, a measurement has to be performed. Measurements are built from the set of \emph{effect operators} $\mathcal{E}(\IC^d):=\lbrace E \in\IC^{d\times d}\ |\ 0\leq E\leq \mathds{1}_d\rbrace$. For our purposes it suffices to consider a measurement as a collection $\lbrace E_i\rbrace_{i=1}^\ell$ of effect operators $E_i\in\mathcal{E}(\IC^d)$ s.t.~$\sum_{i=1}^\ell E_i=\mathds{1}_d$. (For the more general notion of a POVM see \cite{Nielsen.2009} or \cite{Heinosaari.2012}.) When performing a measurement $\lbrace E_i\rbrace_{i=1}^\ell$ on a state $\rho$, output $i$ is observed with probability $\tr[E_i\rho]$. A projective measurement is one where the effect operators are rank-$1$ projections, i.e., there exists an orthonormal basis $\lbrace |i\rangle\rbrace_{i=1}^d$ s.t.~$E_i=|i\rangle\langle i|$.\\

When multiple quantum systems with spaces $\IC^{d_i}$ are considered, the composite system is described by the tensor product $\bigotimes_{i=1}^n \IC^{d_i}\simeq \IC^{\prod_i d_i}$ and the set of states becomes $\mathcal{S} ( \bigotimes_{i=1}^n \IC^{d_i} )$. Given a state $\rho_{AB}\in \mathcal{S}(\IC^{d_A}\otimes\IC^{d_B})$ of a composite system, we can obtain states of the subsystems as partial traces $\rho_A = \tr_B[\rho_{AB}]$, $\rho_B=\tr_A[\rho_{AB}]$. Here, the partial trace is defined as satisfying the relation $\tr[(E\otimes\mathds{1}_{d_B})\rho_{AB}]=\tr[E \tr_B[\rho_{AB}]]$ for all $E\in\mathcal{E}(\IC^{d_A})$.\\

The dynamics of a quantum system are usually described by unitary evolution or, more generally, by quantum channels. For our purposes, these dynamics will not have to be discussed explicitly since they can be considered as part of the performed measurement by changing to the so-called Heisenberg picture (see \cite{Nielsen.2009}). We will take this perspective in proving our sample complexity lower bounds because it allows us to restrict our attention to proving limitations of measurements rather than of channels.\\

We will also make use of some standard entropic quantities which have been generalized from their classical origins \cite{Shannon.1948} to the realm of quantum theory. We denote the Shannon entropy of a random variable $X$ with probability mass function $p$ by $H(X)=-\sum_x p(x)\log(p(x))$, the conditional entropy of a random variable $Y$ given $X$ as $H(Y|X)=\sum_{x,y} p(x,y) \log\left(\tfrac{p(x,y)}{p(x)}\right)$ and the mutual information between $X$ and $Y$ as $I(X:Y)=H(X) + H(Y) - H(X,Y)$. Similarly, the von Neumann entropy of a quantum state $\rho$ will be denoted as $S(\rho)=-\tr[\rho\log\rho]$ and the mutual information for a bipartite quantum state $\rho_{AB}$ as $I(\rho_{AB})=I(A:B)=S(\rho_A) + S(\rho_B) - S(\rho_{AB})$. All the standard results and inequalities connected to these quantities which appear in our arguments can be found in \cite{Nielsen.2009} or in \cite{Wilde.2013}.

\subsection{Basics of the PAC Framework and the Binary Classification Problem}\label{SbSctBasicLearning}
The setting of \emph{Probably Approximately Correct (PAC)} learning was introduced by \cite{Vapnik.1971} and \cite{Valiant.1984}. The general setting is as follows: Let $\mathcal{X}, \mathcal{Y}$ be input and output space, respectively, let $\mathcal{F}\subset\mathcal{Y}^\mathcal{X}$ be a class of functions, a \emph{concept class}, and let $\ell:\mathcal{Y}\times\mathcal{Y}\to\IR_+$ be a \emph{loss function}. A learning algorithm (to which $\mathcal{X},\mathcal{Y},\mathcal{F}$ and $\ell$ are known) has access to training data of the form $S=\lbrace (x_i,y_i)\rbrace_{i=1}^m$, where $(x_i,y_i)$ are drawn i.i.d.~from a probability measure $\mu\in\textrm{Prob}(\mathcal{X}\times\mathcal{Y})$. Moreover, the learner is given as input a confidence parameter $\delta\in (0,1)$ and an accuracy parameter $\varepsilon\in (0,1)$. Then a learner must output a hypothesis $h\in\mathcal{Y}^\mathcal{X}$ s.t., with probability $\geq 1-\delta$ w.r.t.~the choice of training data, 
\begin{align}\label{EqLearningRequirement}
\mathbb{E}_{(x,y)\sim\mu}[\ell(y,h(x))] \leq \inf\limits_{f\in\mathcal{F}} \mathbb{E}_{(x,y)\sim\mu}[\ell(y,f(x))] + \varepsilon.
\end{align}
Note that the first term on the right-hand side vanishes if there exists an $f^*\in\mathcal{F}$ s.t. $\mu(x,y)=\mu_1(x)\delta_{y,f^*(x)}\ \forall (x,y)\in\mathcal{X}\times\mathcal{Y}$. In this case, we call the learning problem \emph{realizable}, otherwise we refer to it as \emph{agnostic}.\\
Both in the agnostic and in the realizable scenario, a learning algorithm that always outputs a hypothesis $h\in\mathcal{F}$ is called a \emph{proper learner}, and otherwise it is called \emph{improper}.\\

A quantity of major interest is the number of examples featuring in such a learning problem. Given a learning algorithm $\mathcal{A}$, the smallest $m=m(\varepsilon,\delta)\in\IN$ s.t.~the learning requirement $(\ref{EqLearningRequirement})$ is satisfied with confidence $1-\delta$ and accuracy $\varepsilon$ is called the \emph{sample complexity} of $\mathcal{A}$. The sample complexity of the learning problem is the infimum over the sample complexities of all learning algorithms for the problem. This characterizes, from an information-theoretic perspective, the hardness of a learning problem, but leaves aside questions of computational complexity.\\

The binary classification problem now arises as a special case from the above if we specify the output space $\mathcal{Y}=\lbrace 0,1\rbrace$ and take the loss function to be $\ell(y,\tilde{y})=1-\delta_{y,\tilde{y}}$, the $0$-$1$-loss. This setting is well studied and a characterization of its sample complexity is known. At its core is the following combinatorial parameter:

\begin{definition}\emph{(VC-Dimension \cite{Vapnik.1971})}\label{DffVCDim}\\
Let $\mathcal{F}\subseteq\lbrace 0,1\rbrace^\mathcal{X}$. A set $S=\lbrace x_1,\ldots,x_n\rbrace\subset X$ is said to be shattered by $\mathcal{F}$ if for every $b\in\lbrace 0,1\rbrace^n$ there exists $f_b\in\mathcal{F}$ s.t.~$f_b(x_i)=b_i$ for all  $1\leq i\leq n$.\\
The Vapnik-Chervonenkis (VC) dimension of $\mathcal{F}\subset\lbrace 0,1\rbrace^\mathcal{X}$ is defined to be 
\begin{align*}
\VC(\mathcal{F}):=\sup\lbrace n\in\IN_0~|~\exists S\subset X~\textrm{s.t. } |S|=n~\textrm{and } S~\textrm{is shattered by }\mathcal{F}\rbrace.
\end{align*}
\end{definition} 

The main insight of VC-theory lies in the fact that learnability of a $\lbrace 0,1\rbrace$-valued concept class is equivalent to finiteness of its VC-dimension. Even more, the sample complexity can be expressed in terms of the VC-dimension. This is the content of the following

\begin{theorem}\label{ThmClassSampleComplexity} \emph{(see, e.g., \cite{Blumer.1989,Hanneke.2016,ShalevShwartz.2014,Vershynin.2018})}\\
In the realizable scenario, the sample complexity of binary classification for a function class $\mathcal{F}$ of VC-dimension $d$ is $m=m(\varepsilon,\delta)=\Theta\left( \frac{1}{\varepsilon}\left( d + \log\nicefrac{1}{\delta}\right)\right)$.\\
In the agnostic scenario, the sample complexity of binary classification for a function class $\mathcal{F}$ of VC-dimension $d$ is $m=m(\varepsilon,\delta)=\Theta\left( \frac{1}{\varepsilon^2}\left( d + \log\nicefrac{1}{\delta}\right)\right)$.
\end{theorem}

The proof of the sample complexity upper bound in the agnostic case typically goes via a different complexity measure, the Rademacher complexity, which is then related to the VC-dimension. As this will reappear later on in our analysis, we also recall this definition here.

\begin{definition}\emph{(Rademacher Complexity (see \cite{Bartlett.2002}))}\\
Let $Z$ be some space, $\mathcal{F}\subseteq\IR^\mathcal{Z}$, $z\in\mathcal{Z}^n$. The empirical Rademacher complexity of $\mathcal{F}$ w.r.t.~$z$ is
\begin{align*}
\hat{\mathcal{R}}(\mathcal{F}):=\underset{\sigma\sim U(\lbrace -1,1\rbrace^n)}{\mathbb{E}}\Big[\sup\limits_{f\in\mathcal{F}}\frac{1}{n}\sum\limits_{i=1}^n \sigma_i f(z_i)\Big]=\underset{\sigma\sim U(\lbrace -1,1\rbrace^n)}{\mathbb{E}}\Big[\sup\limits_{f\in\mathcal{F}}\frac{1}{n}\langle\sigma,f(z)\rangle\Big],
\end{align*}
where $U(\lbrace -1,1\rbrace^n)$ denotes the uniform distribution on $\lbrace -1,1\rbrace^n$.\\
If we consider $n$ i.i.d.~random variables $Z_1,...,Z_n$ distributed according to a probability measure~$\mu$ on $\mathcal{Z}$ and write $Z=(Z_1,...,Z_n)$, the Rademacher complexities of $\mathcal{F}$ w.r.t.~$\mu$ are defined to be $\mathcal{R}_n(\mathcal{F}):=\mathbb{E}_{Z\sim\mu^n}\big[\hat{\mathcal{R}}_\mathcal{F}\big]$, $n\in\IN.$
\end{definition}

\section{The Binary Classification Problem with Classical Instances and Quantum Labels}
We introduce a generalization of the classical binary classification problem to the quantum realm by allowing the two labels to be quantum states. Thus let $\sigma_0,\sigma_1\in\mathcal{S}(\IC^n)$ be two (possibly mixed) quantum states, write $\mathcal{D}=\lbrace\sigma_0,\sigma_1\rbrace$. We assume that classical descriptions of these states (i.e., their density matrices) are known to the learning algorithm as well as the fact that only these two quantum labels appear. The class to be learned is now a class of functions $\mathcal{F}\subset\{ f:\mathcal{X}\to\mathcal{D} \}$ and the underlying distribution will be a $\mu\in\textrm{Prob}(\mathcal{X}\times\mathcal{D})$, where $\mathcal{X}$ is some space of classical objects.\\

We now deviate from the standard PAC setting: We assume the training data to be $S=\lbrace (x_i,\rho_i)\rbrace_{i=1}^m$, $m\in\IN$, where the $(x_i,\rho_i)$ are drawn independently according to $\mu$ (in particular, $\rho_i\in\mathcal{D}$ for all $i$). Here, the~$\rho_i$ are the actual quantum states, not classical descriptions of them. Equivalently, we represent an example $(x_i,\rho_i)$ drawn from $\mu$ as the classical-quantum state $\sum_{x, \rho}\mu (x,\rho)|x\rangle\langle x|\otimes\rho$, with $\lbrace |x\rangle\rbrace_{x\in\mathcal{X}}$ orthonormal.\\

Note that this model for the training data differs from the one introduced by \cite{Bshouty.1998}, where the quantum training data consists of copies of a superposition state. Instead, here we assume copies of a mixture of states. This is done mainly for two reasons: First, it allows us to naturally talk about maps with mixed state outputs. Second, it is debatable whether assuming access to superposition examples as in \cite{Bshouty.1998} is justified (see, e.g., section $5$ in \cite{Ciliberto.2018}), and this problem remains when considering maps with quantum outputs. In contrast, the mixtures assumed in our model arise naturally as statistical ensembles of outputs of state preparation procedures, if the parameters of the preparation are chosen according to some (unknown) distribution. In that sense, the form of classical-quantum training data assumed here is both a straightforward generalization of classical training data, given the standard probabilistic interpretation of mixed states, and can (at least in the realizable scenario) be easily imagined to be obtained as outcome of multiple runs of a state preparation experiment with different parameter settings.\\

A quantum learner for $\mathcal{F}$ with confidence $1-\delta$ and accuracy $\varepsilon$ from $m=m(\varepsilon,\delta)$ quantum examples has to output, for every $\mu\in\textrm{Prob}(\mathcal{X}\times\mathcal{D})$, with probability $\geq 1-\delta$ over the choice of training data of size $m$ according to $\mu$, a hypothesis $h\in\mathcal{D}^\mathcal{X}$ s.t.~$R_\mu(h)\leq \inf\limits_{f\in\mathcal{F}}R_\mu(f) + \varepsilon$. As before, we can consider agnostic versus realizable and proper versus improper variants of this learning model.\\
Here, we define the risk of a hypothesis $h\in\mathcal{F}$ w.r.t.~a distribution $\mu\in\textrm{Prob}(\mathcal{X}\times\mathcal{D})$ as
\begin{align*}
R_\mu(h):= \int\limits_{\mathcal{X}\times\mathcal{D}} \frac{1}{2} \lVert\rho - h(x)\rVert_1 ~ \mathrm{d}\mu(x,\rho),
\end{align*}
where $\lVert\rho - \sigma\rVert_1 = \tr[|\rho-\sigma|]=\tr[\sqrt{(\rho-\sigma)^*(\rho-\sigma)}]$ is the Schatten $1$-norm.\\

Note that our assumption on $\mathcal{F}$ implies that $h(x)\in\mathcal{D}\ \forall x\in\mathcal{X}$ and therefore we can easily rewrite $R_\mu (h)=\frac{\norm{\sigma_0-\sigma_1}_1}{2}\mathbb{P}_{(x,\rho)\sim\mu}[h(x)\neq \rho]$, which is just the $0$-$1$-risk multiplied by a constant. We choose the slightly more complicated looking definition for $R_\mu(h)$ for two reasons. On the one hand, $\frac{\norm{\sigma_0-\sigma_1}_1}{2}$ is a measure for the distinguishability of $\sigma_0$ and $\sigma_1$ and thus a natural scale w.r.t.~which to measure the prediction error. (Note: If $\sigma_0,\sigma_1$ are orthogonal pure states and thus perfectly distinguishable, the classical scenario is recovered.) On the other hand, our definition of risk can be motivated operationally as we discuss in Appendix \ref{SctRiskMotivation}.

\begin{example}\label{ExmGroundStatePreparationLearningProblem}
Here, we describe a physically motivated problem that is captured by our scenario. The idea is as follows: Suppose we have available a (possibly complicated) ground state preparation procedure. Using this, we want to prepare a ground state $|\varphi_0\rangle$ of a Hamiltonian $H$. However, $H$ is perturbed by noise about which we have only partial information. We want to learn more about the noise and its influence on the prepared ground state.\\
We make this idea more concrete. We consider a (self-adjoint) Hamiltonian $H\in\mathbb{C}^{(d+2)\times (d+2)}$ of the form $H=\mathds{1}_2 \oplus \tilde{H}$, where $\tilde{H}>\mathds{1}_d$, with (non-unique) ground state $|\varphi_0\rangle:=\begin{pmatrix} 0 & 1\end{pmatrix}^T\oplus 0$. Suppose that we have a ground state preparation procedure that, if run with Hamiltonian $H$, prepares $|\varphi_0\rangle$. When implementing this procedure, we have to fix values of a parameter vector $x\in\mathbb{R}^D$. (Think, e.g., of $D=3$ and $x$ denoting the location at which the experiment is set up.) But due to the laboratory being only imperfectly shielded, there is an unknown region $R\subset\mathbb{R}^D$ in which the system is subject to noise. For simplicity, we assume that only two types of noise can occur and lead to the location-dependent Hamiltonian $H^{(i)}_x = H + \mathds{1}_{\{x\in R\}} H^{(i)},$ with noise Hamiltonians $H^{(0)} = \begin{pmatrix} 1 & 0 \\ 0 & -1 \end{pmatrix}\oplus 0$, $H^{(1)} = \begin{pmatrix} 0 & 1 \\ 1 & 0\end{pmatrix}\oplus 0$.\\
The noise can lead to a perturbation of the ground state. Namely:
\begin{itemize}
    \item For $x\not\in R$, $|\varphi_0\rangle$ is a ground state of $H^{(i)}_x$. (This is the case of no effective noise.)
    \item For $x\in R$, $|\varphi_0\rangle$ is the unique ground state of $H^{(0)}_x$. Hence, the noise $H^{(0)}$ is benign from the perspective of ground state preparation. 
    \item For $x\in R$, $|\varphi_1\rangle:=\frac{1}{\sqrt{2}}\begin{pmatrix} 1 & -1 \end{pmatrix}^T\oplus 0$ is the unique ground state of $H^{(1)}_x$. Hence, the noise $H^{(1)}$ is malicious from the perspective of ground state preparation. 
\end{itemize}
Thus, we describe the ground state preparation by a function $f^{(i)}_R:\mathbb{R}^D\to\{|\varphi_0\rangle\langle\varphi_0|, |\varphi_1\rangle\langle\varphi_1|\}$, $f^{(i)}_R(x) = \mathds{1}_{\{x\not\in R\}}|\varphi_0\rangle\langle\varphi_0| + \mathds{1}_{\{x\in R\}}|\varphi_i\rangle\langle\varphi_i|$. With this formulation, gaining information about the noise region $R$ and the noise type $i$ can be phrased as the problem of (PAC-)learning an unknown element of the (known) function class $\mathcal{F}=\left\{f^{(i)}_R\right\}_{i=0,1,~R\in\mathcal{R}}\subseteq\{|\varphi_0\rangle\langle\varphi_0|, |\varphi_1\rangle\langle\varphi_1|\}^{\mathbb{R}^D}$, where $\mathcal{R}$ is the class of possible error regions.\\
Note that $|\varphi_0\rangle$ and $|\varphi_1\rangle$ are not orthogonal and thus cannot be perfectly distinguished. Therefore, we cannot phrase the learning problem as one of binary classification with classical labels.\\
We return to this setting in Examples \ref{ExmGroundStatePreparationAgnostic} and \ref{ExmGroundStatePreparationRealizable} to illustrate our learning strategies.

\end{example}

We want to conclude this section by discussing a drawback of our model. We assume $\mathcal{F}\subset\mathcal{D}^\mathcal{X}$, i.e., outputs of any $f\in\mathcal{F}$ are either $\sigma_0$ or $\sigma_1$. Considering the convex structure of the set of quantum states, which is intimately tied to the probabilistic interpretation of quantum theory, this restriction can be considered unnatural. We nevertheless make it, for two reasons: First, it is easy to show using a Bayesian predictor that, under the assumption of $\mu$ being supported on $\mathcal{D}$ (which could, of course, also be contested), the optimal choice of predictors among all functions $\left(\mathcal{S}(\IC^d)\right)^\mathcal{X}$ is actually a function in $\mathcal{D}^\mathcal{X}$. Second, it is the most direct analog of the classical scenario with binary labels and we consider it a sensible first step that, as demonstrated in Example \ref{ExmGroundStatePreparationLearningProblem}, can already be of physical relevance.

\section{Sample Complexity Upper Bounds}\label{SctUpperBounds}

\subsection{The Agnostic Case}\label{SbSctAgnosticUpperBound}
Our learning strategy is motivated by interpreting the classical training data arising from performing a measurement on the label states as noisy version of the true training data. Before describing the learning strategy, we recall our assumption that classical descriptions of the label states $\sigma_0$, $\sigma_1$ are known to the learner. Based on this knowledge, the learner can derive the optimal measurement $\lbrace E_0,E_1\rbrace$ for minimum-error distinction between the two states, the so-called Holevo-Helstrom measurement (see Theorem $3.4$ in \cite{Watrous.2018}), by choosing $E_0$ to be the orthogonal projector onto the eigenspaces of $\sigma_0-\sigma_1$ corresponding to nonnegative eigenvalues. This step is where knowledge of the states $\sigma_0$ and $\sigma_1$ is used.\\

The learning strategy is now the following, in which we use the Holevo-Helstrom measurement to produce classical training data and thus obtain a classical learning problem:

\begin{tcolorbox}[title={Noise-corrected Holevo-Helstrom strategy}, breakable]
Given: Quantum training data $S=\{ (x_i,\rho_i)\}_{i=1}^m$\\
Output: Hypothesis $\hat{h}:\mathcal{X}\to\mathcal{D}$\\
Algorithm: 
\begin{enumerate}
\item For each $i$: Perform a Holevo-Helstrom measurement on~$\rho_i$. Let~$$y_i =\begin{cases} 1\ &\textrm{if } E_1~ \textrm{is accepted}\\ 0 &\textrm{if } E_1~ \textrm{is rejected}\end{cases}.$$
\item Let $\tilde{S}=\lbrace (x_i,y_i)\rbrace_{i=1}^m\in (\mathcal{X}\times\lbrace 0,1\rbrace)^m$. Then one can view $(x_i,y_i)$ as being drawn independently according to the probability measure $\nu$ on $\mathcal{X}\times\lbrace 0,1\rbrace$ which has 
\begin{align*}
&\nu_1(x) = \mu_1 (x) = \mu (x,\sigma_0) + \mu (x,\sigma_1)
\end{align*}
as first marginal and 
\begin{align*}
\nu (y|x) = \ &\delta_{y0}\left(\mu (\sigma_1 |x) \tr[\sigma_1 E_0] + \mu (\sigma_0 |x) \tr[\sigma_0 E_0]\right)\\
		      +~ &\delta_{y1}\left(\mu (\sigma_1 |x) \tr[\sigma_1 E_1] + \mu (\sigma_0 |x) \tr[\sigma_0 E_1]\right).	
\end{align*}
as the conditional probability distribution of $y$ given $x$.
\item Use a classical learning algorithm to find $\hat{g}\in\tilde{\mathcal{F}}:=\lbrace \tilde{f}:\mathcal{X}\to\lbrace 0,1\rbrace\ |\ \exists f\in\mathcal{F}:\ f(x)=\sigma_{\tilde{f}(x)}\ \forall x\in\mathcal{X}\rbrace$ s.t.~$\tilde{R}_\nu (\hat{g}):= \mathbb{E}_{(x,y)\sim\nu} [\tilde{\ell}(y,\hat{g}(x)]$ is minimized over $\tilde{\mathcal{F}}$, where
\begin{align*}
\tilde{\ell}(y_1,y_2):=\frac{(1-\eta_{1\oplus y_2})\mathds{1}_{y_1\neq y_2} - \eta_{y_2}\mathds{1}_{y_1=y_2}}{1-\eta_0-\eta_1},
\end{align*}
with $\eta_0 = \tr[\sigma_0 E_1]$ and $\eta_1 = \tr[\sigma_1 E_0]$. Here, $\oplus$ denotes addition modulo $2$.
\item Define $\hat{h}:\mathcal{X}\to\mathcal{D}$ via $\hat{h}(x)=\sigma_{\hat{g}(x)}$ and output $\hat{h}$ as hypothesis.
\end{enumerate}
\end{tcolorbox}

Note that the only non-classical step in the strategy is step $(1)$, which consists only of performing local two-outcome measurements.\\

The modification of the loss function in step $(3)$ gives an unbiased estimate of the true risk:

\begin{lemma}\emph{(see Lemma $1$ in \cite{Natarajan.2013})}\label{LmmFaithfulRiskEstimate}\\
Fix $x\in\mathcal{X}$. With the notation introduced above, for every $z\in\lbrace 0,1\rbrace$ it holds that
\begin{align*}
\mathbb{E}_{Y\sim\nu(\cdot|x)}[\tilde{\ell}(z,Y)] = &\mathbb{E}_{Y\sim\mu(\cdot|x)}[\mathds{1}_{z\neq Y}].
\end{align*}
\end{lemma}

We can use a standard generalization bound in terms of Rademacher complexities (see, e.g., Theorem $26.5$ of \cite{ShalevShwartz.2014}) to obtain: With probability $\geq 1-\delta$ over the choice of training data $S=\{(x_i,y_i) \}_{i=1}^m$ according to $\nu$, we have that for all $\tilde{f}^\ast\in\mathcal{\tilde{F}}$
\begin{align*}
\mathbb{E}_{(x,y)\sim\nu} [\tilde{\ell}(\hat{g}(x),y)] - \mathbb{E}_{(x,y)\sim\nu} [\tilde{\ell}(\tilde{f}^\ast(x),y)]
\leq 2\hat{\mathcal{R}}(\tilde{\mathcal{G}}) + \frac{5}{1-\eta_0-\eta_1}\sqrt{\frac{2\ln\nicefrac{8}{\delta}}{m}},
\end{align*}
where we used that $\lvert\tilde{\ell}(y_1,y_2)\rvert\leq \frac{1}{1-\eta_0-\eta_1}$ and defined the function class 
\begin{align*}
\tilde{\mathcal{G}} := \{ \mathcal{X}\times\{ 0,1\}\ni (x,y)\mapsto \tilde{\ell}(\tilde{f}(x),y)~|~\tilde{f}\in\tilde{\mathcal{F}}\}.
\end{align*}

Next, we relate the empirical Rademacher complexity of $\tilde{\mathcal{G}}$ to that of $\tilde{\mathcal{F}}$.

\begin{lemma}\label{LmmRademacherComplexityNoiseCorrectedClass}
For any training data set $S=\{(x_i,y_i) \}_{i=1}^m$, viewed as an element of $(\mathcal{X}\times\{ 0,1\})^m$, we have
\begin{align*}
\hat{\mathcal{R}} (\tilde{\mathcal{G}})
\leq \frac{2}{1-\eta_0-\eta_1}\hat{\mathcal{R}}(\tilde{\mathcal{F}}).
\end{align*}
\end{lemma}
\begin{proof} (Sketch)
The proof uses some standard steps that are typically used for example in proving the Lipschitz contraction property of the Rademacher complexity and in studying the Rademacher complexity in a binary classification scenario.\\
See Appendix \ref{SctProofs} for a detailed proof.
\end{proof}
\\

With this, we now reformulate the above result in terms of the VC-dimension. Suppose $\VC (\tilde{\mathcal{F}})=d<\infty$. Then $\hat{\mathcal{R}}(\tilde{\mathcal{F}}) \leq 31\sqrt{\frac{d}{m}}$ (see, e.g., Theorem $8.3.23$ in \cite{Vershynin.2018}). Therefore we obtain that, with probability $\geq 1-\delta$ over the choice of training data $S=\{(x_i,y_i) \}_{i=1}^m$ according to $\nu$, 
\begin{align*}
\mathbb{E}_{(x,y)\sim\nu} [\tilde{\ell}(\hat{g}(x),y)] - \inf\limits_{\tilde{f}\in\tilde{\mathcal{F}}}\mathbb{E}_{(x,y)\sim\nu} [\tilde{\ell}(\tilde{f}(x),y)]
\leq \frac{124}{1-\eta_0-\eta_1}\sqrt{\frac{d}{m}} + \frac{5}{1-\eta_0-\eta_1}\sqrt{\frac{2\ln\nicefrac{8}{\delta}}{m}}.
\end{align*}

Note that, using Lemma \ref{LmmFaithfulRiskEstimate}, we can now bound 
\begin{align*}
R_\mu(\hat{h}) - \inf\limits_{f\in\mathcal{F}} R_\mu (f)
&= \frac{\lVert\sigma_0-\sigma_1\rVert_1}{2}\underbrace{\mathbb{E}_{(x,\rho)\sim\mu} [\mathds{1}_{\hat{g}(x)\neq \rho}]}_{=\mathbb{E}_{(x,y)\sim\nu} [\tilde{\ell}(\hat{g}(x),y)]} - \inf\limits_{\tilde{f}\in\tilde{\mathcal{F}}}\frac{\lVert\sigma_0-\sigma_1\rVert_1}{2}\underbrace{\mathbb{E}_{(x,\rho)\sim\mu} [\mathds{1}_{\tilde{f}(x)\neq \rho}}_{=\mathbb{E}_{(x,y)\sim\nu} [\tilde{\ell}(\tilde{f}(x),y)]}]\\
&\leq \frac{\lVert\sigma_0-\sigma_1\rVert_1}{2}\left(\frac{124}{1-\eta_0-\eta_1}\sqrt{\frac{d}{m}} + \frac{5}{1-\eta_0-\eta_1}\sqrt{\frac{2\ln\nicefrac{8}{\delta}}{m}} \right).
\end{align*}

Now we can set this equal to $\varepsilon$ and rearrange to conclude that a sample size of 
\begin{align*}
m\geq \frac{\lVert\sigma_0-\sigma_1\rVert_1^2}{4\varepsilon^2}\left(\frac{124}{1-\eta_0-\eta_1}\sqrt{d} + \frac{5}{1-\eta_0-\eta_1}\sqrt{2\ln\nicefrac{8}{\delta}} \right)^2
\end{align*}
suffices to guarantee that, with probability $\geq 1-\delta$, $R_\mu(\hat{h}) - \inf\limits_{f\in\mathcal{F}} R_\mu (f)\leq\varepsilon$. If we now observe that $\frac{1}{1-\eta_0-\eta_1}\leq\frac{4}{\lVert\sigma_0-\sigma_1\rVert_1}$, we obtain the sample complexity upper bound
\begin{align*}
m =m(\varepsilon,\delta) = \mathcal{O}\left( \frac{d}{\varepsilon^2} + \frac{\log\nicefrac{1}{\delta}}{\varepsilon^2}\right).
\end{align*}

\begin{remark}
The naive version of our learning strategy would be to perform Holevo-Helstrom measurements and then apply a classical learning strategy, like empirical risk minimization, without correcting for the noise in the resulting classical labels. Actually, this learning strategy already performs reasonably well and, in certain special cases, even allows to reduce the quantum learning problem to a fully classical one. For a detailed analysis of the performance of this simpler strategy, the reader is referred to Appendix \ref{SctHelstrom}.
\end{remark}

\begin{example}\label{ExmGroundStatePreparationAgnostic}
We illustrate our agnostic learning strategy for the scenario of Example \ref{ExmGroundStatePreparationLearningProblem}. As discussed in Appendix \ref{SctHelstrom}, as both label states $|\varphi_0\rangle\langle\varphi_0|$ and $|\varphi_1\rangle\langle\varphi_1|$ are pure, we can actually dispense with the modification of the classical loss function and simply take the $0$-$1$-loss. Therefore, the Holevo-Helstrom strategy will look as follows: We first perform local Holevo-Helstrom measurements with measurement operators $E_0 \propto \begin{pmatrix} -1+\sqrt{2} & 1\end{pmatrix}^T \begin{pmatrix} -1+\sqrt{2} & 1\end{pmatrix}\oplus 0$, $E_1 = \mathds{1}_{2+d}-E_0$. This gives rise to classical training data. With that data, we then perform (classical) empirical risk minimization over the class $\tilde{\mathcal{F}}=\left\{\tilde{f}^{(i)}_R \right\}_{i=0,1,~R\in\mathcal{R}}$, where $\tilde{f}^{(i)}_R:\mathbb{R}^D\to\{0,1\}$, $\tilde{f}^{(i)}_R(x)=\mathds{1}_{\{x\in R\}}\delta_{i,1}$. Note that $f^{(0)}_R$ is the zero-function for every $R\in\mathcal{R}$.\\
Both the optimization procedure and the generalization capability depend on the class $\mathcal{R}$ of possible noise regions. Concerning the generalization performance, observerve that, if $\emptyset\in\mathcal{R}$, then $\VC ( \tilde{\mathcal{F}})=\VC(\tilde{\mathcal{F}}_\mathcal{R})$, where we take $\tilde{\mathcal{F}}_\mathcal{R}=\{\mathbb{R}^D\ni x\mapsto\mathds{1}_{\{x\in R\}}~|~R\in\mathcal{R}\}$ to be the class of indicator functions of sets from $\mathcal{R}$. The VC-dimension of such classes is well-known for different geometric classes $\mathcal{R}$. E.g., if $\mathcal{R}$ is the class of axis-aligned rectangles or that of Euclidean balls in $\mathbb{R}^D$, then $\VC (\tilde{\mathcal{F}}_\mathcal{R})$ scales linearly in $D$ and thus the dependence of the sample complexity upper bound on the number of parameters $D$ is linear. If, however, we take $\mathcal{R}$ to be the class of compact and convex subsets of $\mathbb{R}^D$, then $\VC (\tilde{\mathcal{F}}_\mathcal{R})=\infty$ and the sample complexity upper bound becomes void. This is congruent with the intuition that without prior assumptions on the structure of the regions that can be influenced by noise, learning the noise (in particular its region) will be hard and maybe infeasible. 
\end{example}

\subsection{The Realizable Case}\label{SbSctRealizableUpperBound}

The strategy from the previous subsection uses a generalization bound via the Rademacher complexity and yields a sample complexity bound depending quadratically on $\nicefrac{1}{\varepsilon}$. In the classical binary classification problem it is known (see Theorem \ref{ThmClassSampleComplexity}) that under the realizability assumption this can be improved to $\nicefrac{1}{\varepsilon}$, but this typically requires a different kind of reasoning via $\varepsilon$-nets. (Compare section $28.3$ of \cite{ShalevShwartz.2014}.) In Theorem \ref{NoisySampleComplexityUpperBound} we show how the reasoning by \cite{Hanneke.2016} can be combined with results by \cite{Laird.1988} to achieve the $\nicefrac{1}{\varepsilon}$-scaling also in the case of two-sided classification noise. This sample complexity upper bound is seen to be optimal in its dependence on the VC-dimension $d$, the error rate bound $\eta$, the confidence $\delta$ and the accuracy $\varepsilon$ by a comparison to the lower bound in Theorem $27$ of \cite{Arunachalam.2018}.\\

If, as in the previous subsection, we consider the classical training data obtained by measuring the quantum training data as noisy version of a true sample, we can exchange step $3$ in the Holevo-Helstrom strategy by the minimum-disagreement-based classical learning strategy achieving the optimal sample complexity bound of Theorem \ref{NoisySampleComplexityUpperBound}. This directly yields the following

\begin{theorem}\label{ThmRealizableSampleComplexityUpperBound}
Let $\sigma_0,\sigma_1\in\mathcal{S}(\IC^n)$ be (distinct) quantum states. Let $\varepsilon\in (0,1)$, $\delta\in (0,2\cdot (\tfrac{2e}{d})^d)$, where $d$ is the VC-dimension of $\mathcal{F}\subset \lbrace 0,1\rbrace^\mathcal{X}$. Then $$m=m(\varepsilon,\delta)=\mathcal{O}\left(\frac{1}{\varepsilon (1-2\max\lbrace \tr[E_0\sigma_1],\tr[E_1\sigma_0]\rbrace)^2} \left( d + \log\nicefrac{1}{\delta}\right)\right)$$ quantum examples of a function in $\mathcal{F}$ are sufficient for binary classification with classical instances and quantum labels $\sigma_0,\sigma_1$ with accuracy $\varepsilon$ and confidence $1-\delta$.
\end{theorem}

\begin{example}\label{ExmGroundStatePreparationRealizable}
When considering this learning strategy in the setting of Example \ref{ExmGroundStatePreparationLearningProblem}, we first perform the Holevo-Helstrom measurements as in Example \ref{ExmGroundStatePreparationAgnostic} to obtain classical data. Again, this is followed by a classical learning procedure for the class $\tilde{\mathcal{F}}=\left\{\tilde{f}^{(i)}_R \right\}_{i=0,1,~R\in\mathcal{R}}$.\\
Whereas the sample complexity bound derived for the agnostic case in subsection \ref{SbSctAgnosticUpperBound} applies to any (noise-corrected) classical empirical risk minimization, the procedure leading to the bound in Theorem \ref{ThmRealizableSampleComplexityUpperBound} is a specific one, presented in the proof of Theorem \ref{NoisySampleComplexityUpperBound}. First, the classical data is processed, using the subsampling algorithm of \cite{Hanneke.2016} (see Algorithm \ref{SubsampleAlgo}), to generate a collection of subsamples. For each of those subsamples, we then apply Algorithm \ref{MinDisagreementAlgo}: We use a first part of the subsample to group the elements of $\tilde{\mathcal{F}}$ into equivalence classes (according how they act on that part of the subsample), and the remainder is used to test the performance of each equivalence class. Afterwards, we output as hypothesis for that subsample a representative of the equivalence class that performs best in that test, i.e., that minimizes the number of disagreements with the part of the subsample used for testing. Whether and how the grouping into equivalence classes and finding minimum-disagreement strategies can be done (efficiently) depends on $\tilde{\mathcal{F}}$, and thus on $\mathcal{R}$. Finally, we take a majority vote over all the subsample hypotheses to get the output hypothesis of the classical learning procedure.\\
The dependence of the sample complexity on $\tilde{\mathcal{F}}$ via the VC-dimension of the class of indicator functions of sets from $\mathcal{R}$ is analogous to Example \ref{ExmGroundStatePreparationAgnostic}.
\end{example}

\begin{remark}
From the description of our noise-corrected Holevo-Helstrom strategy (either in the form of subsection \ref{SbSctAgnosticUpperBound} or that of this subsection), we can directly see that whether it is a proper or an improper learner depends on whether the classical learning algorithm in step $(3)$ is. As the classical learning algorithm used in subsection \ref{SbSctAgnosticUpperBound} is a simple Empirical Risk Minimization, it is in particular proper. So our noise-corrected Holevo-Helstrom strategy for the agnostic case is proper as well. The classical learner used in this subsection, however, is in general improper. So also the noise-corrected Holevo-Helstrom strategy for the realizable case is in general improper.
\end{remark}

\section{Sample Complexity Lower Bounds}\label{SctLowerBounds}
Whereas the goal of the previous section was to give strategies for solving the binary classification problem with classical instances and quantum labels and to prove upper bounds on the sufficient number of classical-quantum examples, we now turn to the complementary question of lower bounds on the number of required examples. In this section, we derive lower bounds that match the respective upper bounds from the previous section and therefore we conclude that the procedures described in Section \ref{SctUpperBounds} are optimal w.r.t.~sample size in terms of the dependence on $\varepsilon$, $\delta$, and $d$.

\subsection{The Agnostic Case}
We prove the sample complexity lower bounds in two parts, the first depending on the confidence parameter $\delta$ but not on the VC-dimension of the function class and conversely for the second.\\

We establish the VC-dimension-independent sample complexity lower bound in the following

\begin{lemma}\label{LmmAgnosticConfidenceLowerBound}
Let $\sigma_0,\sigma_1\in\mathcal{S}(\IC^n)$, let $\varepsilon\in (0,\tfrac{\norm{\sigma_0-\sigma_1}_1}{2\sqrt{2}})$, $\delta\in (0,1)$. Let $\mathcal{F}\subset\mathcal{D}^\mathcal{X}$ be a non-trivial concept class. Suppose $\mathcal{A}$ is a learning algorithm that solves the binary classification task with classical instances and (distinct) label states $\sigma_0,\sigma_1$ and concept class $\mathcal{F}$ with confidence $1-\delta$ and accuracy $\varepsilon$ using $m=m(\varepsilon,\delta)$ examples. Then $m\geq \Omega\left(\norm{\sigma_0-\sigma_1}_1^2\frac{\log\nicefrac{1}{\delta}}{\varepsilon^2}\right)$.
\end{lemma}
\begin{proof} (Sketch)
As $\mathcal{F}$ is non-trivial, there exist concepts $f, g\in\mathcal{F}$ and a point $x\in\mathcal{X}$ s.t.~$f(x)=\sigma_0$ and $g(x)=\sigma_1$. Let $\lambda=\frac{\varepsilon}{2\norm{\sigma_0 - \sigma_1}_1}\in (0,1)$. Define probability distributions $\mu_{\pm}$ on $\mathcal{X}\times\mathcal{D}$ via
\begin{align*}
\mu_{\pm}(x,f(x)) = \frac{1\pm \lambda}{2},\quad \mu_{\pm}(x,g(x))=\frac{1\mp\lambda}{2}.
\end{align*}
By explicitly evaluating the risk $R_{\pm}(h)$, we see that achieving an excess risk $\leq\varepsilon$ with probability $\geq 1-\delta$, requires the learner to distinguish between the underlying distributions $\mu_\pm$, and thus the corresponding training data states $\rho_{\pm}^{\otimes m}$, with probability $\geq 1-\delta$.\\
It is well known (see, e.g., \cite{Nielsen.2009}, chapter $9$) that the optimal success probability of this quantum distinguishing task is given by $p_{\textrm{opt}} = \frac{1}{2}(1+\frac{1}{2}\norm{\rho_+^{\otimes m} - \rho_-^{\otimes m}}_1).$ Via the Fuchs-van de Graaf inequalities, which state that $$\frac{1}{2}\norm{\rho_1^{\otimes m} - \rho_2^{\otimes m}}_1 \leq \sqrt{1-F(\rho_1^{\otimes m}, \rho_2^{\otimes m})^2}
= \sqrt{1-F(\rho_1, \rho_2)^{2m}},$$ this can be upper-bounded using lower bounds on the fidelity $F(\rho_+^{\otimes m}, \rho_-^{\otimes m})=F(\rho_+, \rho_-)^m$. The fidelity $F(\rho_+, \rho_-)$ can be lower-bounded using its strong concavity and the explicit expressions for $\rho_\pm$. The result then follows by comparing the obtained upper bound with the required lower bound $p_{\textrm{opt}}\geq 1-\delta$.\\
See Appendix \ref{SctProofs} for a detailed proof.
\end{proof}
\\

For the proof of the VC-dimension-dependent part of the lower bound we need a well known observation about the eigenvalues of a statistical mixture of two pure quantum states, which is the content of the following

\begin{lemma}\label{LmmEigenvaluesStatisticalMixture}
Let $|\psi\rangle,|\phi\rangle\in\mathbb{C}^n$ be distinct pure quantum states. Let $\alpha,\beta\geq 0$ be real numbers. Then the non-zero eigenvalues of the mixture $\rho:=\alpha |\psi\rangle\langle\psi| + \beta |\phi\rangle\langle\phi|$ are given by
\begin{align*}
\lambda_{1/2}(\rho) = \frac{\alpha+\beta\pm\sqrt{(\alpha-\beta)^2 + 4\alpha\beta |\langle\psi |\phi\rangle|^2}}{2}.
\end{align*}
\end{lemma}

With this we can now prove a sample complexity lower bound for the case of pure label states.

\begin{theorem}\label{ThmAgnosticVCDimLowerBound}
Let $\sigma_0=|\psi_0\rangle\langle\psi_0|,\sigma_1=|\psi_1\rangle\langle\psi_1|\in\mathcal{S}(\IC^n)$ be (distinct) pure quantum states, let $\varepsilon\in (0,\frac{\norm{\sigma_0-\sigma_1}_1}{8})$, $\delta\in (0,1-H\left(\tfrac{1}{4}\right))$. Let $\mathcal{F}\subset\mathcal{D}^\mathcal{X}$ be a non-trivial concept class s.t.~$\tilde{\mathcal{F}}$ has VC-dimension $d$. Suppose $\mathcal{A}$ is a learning algorithm that solves the binary classification task with classical instances and (distinct) label states $\sigma_0,\sigma_1$ and concept class $\mathcal{F}$ with confidence $1-\delta$ and accuracy $\varepsilon$ using $m=m(\varepsilon,\delta)$ examples. Then $m\geq \Omega\left(\frac{d}{\varepsilon^2}\right)$.
\end{theorem}
\begin{proof} (Sketch)
We follow the information-theoretic proof strategy from \cite{Arunachalam.2018}. Let $S=(s_1,\ldots,s_d)\in\mathcal{X}$ be a set shattered by $\tilde{\mathcal{F}}$, for each $a\in\lbrace 0,1\rbrace^d$ define the distribution $\mu_a$ on $\lbrace 1,\ldots,d\rbrace\times\lbrace 0,1\rbrace$ via
\begin{align*}
\mu_a(i,b) := \frac{1}{2d}\left( 1 + (-1)^{a_i + b} \frac{8\varepsilon}{\norm{\sigma_0-\sigma_1}_1}\right).
\end{align*}
Note that $\forall a\in\lbrace 0,1\rbrace^d\ \exists f_a\in\tilde{\mathcal{F}}: f_a(s_i)=a_i$ by shattering and that $f_a$ is a minimum-error concept w.r.t.~$\mu_a$. By evaluating the excess error of an $f_{\tilde{a}}$ compared to $f_a$, we see that solving the learning problem with confidence $1-\delta$ requires the learner to output, with probability $\geq 1-\delta$, a hypothesis described by a string whose Hamming distance to the true underlying string is $\leq \frac{d}{4}$. We can use this observation to obtain the lower bound $I(A:B)\geq \Omega(d)$ on the mutual information between underlying string $A$ (drawn uniformly at random) and corresponding quantum training data $B$.\\
We can also upper-bound the mutual information. A standard argument shows $I(A:B)\leq m\cdot I(A:B_1)$, where $m$ is the number of copies of the quantum example state and $B_1$ describes a single quantum example state. Using Lemma \ref{LmmEigenvaluesStatisticalMixture} and the explicit expression for a quantum example state, we can compute $I(A:B_1)$ and use Taylor expansion to see that $I(A:B_1)\leq \mathcal{O}(\varepsilon^2)$. Comparing the lower and upper bounds on $I(A:B)$ now gives $m\geq \Omega\left(\frac{d}{\varepsilon^2}\right)$.\\
See Appendix \ref{SctProofs} for a detailed proof.
\end{proof}
\\

If we now combine Lemma \ref{LmmAgnosticConfidenceLowerBound} and Theorem \ref{ThmAgnosticVCDimLowerBound} with the result of Subsection \ref{SbSctAgnosticUpperBound} we obtain

\begin{corollary}
Let $\sigma_0,\sigma_1\in\mathcal{S}(\IC^n)$ be (distinct) pure quantum states, let $\varepsilon\in (0,\frac{\norm{\sigma_0-\sigma_1}_1}{8})$, $\delta\in (0,1-H\left(\tfrac{1}{4}\right))$. Let $\mathcal{F}\subset\mathcal{D}^\mathcal{X}$ be a non-trivial concept class s.t.~$\tilde{\mathcal{F}}$ has VC-dimension $d$. Then a sample size of $\Theta\left( \frac{d}{\varepsilon^2} + \frac{\log\nicefrac{1}{\delta}}{\varepsilon^2} \right)$ is necessary and sufficient for solving the binary classification task with classical instances and quantum labels $\sigma_0,\sigma_1$ and hypothesis class $\mathcal{F}$ with confidence $1-\delta$ and accuracy $\varepsilon$.
\end{corollary}

Therefore we have shown that the strategy from Subsection \ref{SbSctAgnosticUpperBound} is, for pure states, optimal in sample complexity w.r.t.~its dependence the VC-dimension, the accuracy and the confidence. But we do not make a statement on optimality w.r.t.~the dependence on the distinguishability of the label states, because the parameter $\norm{\sigma_0 - \sigma_1}_1$ is lacking from our lower bound.

\subsection{The Realizable Case}
We now show analogous lower bounds for the sample complexity in the realizable scenario with the same proof strategy.

\begin{lemma}\label{LmmRealizableConfidenceLowerBound}
Let $\sigma_0,\sigma_1\in\mathcal{S}(\IC^n)$, let $\varepsilon\in (0,\tfrac{\norm{\sigma_0-\sigma_1}_1}{2})$, $\delta\in (0,\frac{1}{2})$. Let $\mathcal{F}\subset\mathcal{D}^\mathcal{X}$ be a non-trivial concept class. Suppose $\mathcal{A}$ is a learning algorithm which solves the binary classification task with classical instances and (distinct) label states $\sigma_0,\sigma_1$ and concept class $\mathcal{F}$ with confidence $1-\delta$ and accuracy $\varepsilon$ using $m=m(\varepsilon,\delta)$ examples in the realizable scenario. Then $m\geq \Omega\left(\frac{\log\nicefrac{1}{\delta}}{\varepsilon}\right)$.
\end{lemma}
\begin{proof}
This can be proved similarly to Lemma \ref{LmmAgnosticConfidenceLowerBound}. See Appendix \ref{SctProofs} for a detailed proof.
\end{proof}

We now provide the analog of Theorem \ref{ThmAgnosticVCDimLowerBound} for the realizable case.

\begin{theorem}\label{ThmRealizableVCDimLowerBound}
Let $\sigma_0=|\psi_0\rangle\langle\psi_0|,\sigma_1=|\psi_1\rangle\langle\psi_1|\in\mathcal{S}(\IC^n)$ be (distinct) pure quantum states, let $\varepsilon\in (0,\frac{\norm{\sigma_0-\sigma_1}_1}{8})$, $\delta\in (0,\frac{1}{2})$. Let $\mathcal{F}\subset\mathcal{D}^\mathcal{X}$ be a non-trivial concept class s.t.~$\tilde{\mathcal{F}}$ has VC-dimension $d+1$. Suppose $\mathcal{A}$ is a learning algorithm which solves the binary classification task with classical instances and (distinct) label states $\sigma_0,\sigma_1$ and concept class $\mathcal{F}$ with confidence $1-\delta$ and accuracy $\varepsilon$ using $m=m(\varepsilon,\delta)$ examples in the realizable case. Then $m\geq \Omega\left(\frac{d}{\varepsilon}\right)$.
\end{theorem}
\begin{proof}
This can be proved similarly to Theorem \ref{ThmAgnosticVCDimLowerBound} . See Appendix \ref{SctProofs} for a detailed proof.
\end{proof}
\\

Thus, we have obtained a sample complexity lower bound that matches the upper bound proved in Subsection \ref{SbSctRealizableUpperBound} in the dependence on the VC-dimension, the confidence and the accuracy, but we do not make a statement about optimality w.r.t.~the dependence on $\norm{\sigma_0-\sigma_1}_1$.

\begin{remark}
As already discussed in subsection \ref{SbSctBasicQIT}, in proving the sample complexity lower bounds we resort to the Heisenberg picture, which allows us to absorb the intermediate quantum channels performed by a learner into the measurement. These lower bounds therefore even hold for quantum learning algorithms that perform coherent and adaptive measurements on the training data. In particular, the information-theoretic complexity of our learning problem does not change if we restrict the quantum learner to only performing two-outcome POVMs locally (i.e., on one subsystem only). This is maybe not too much of a surprise, since the optimal measurement for distinguishing states drawn uniformly at random from $\{\bigotimes_{i=1}^m \sigma_{x_i}\}_{x\in\{0,1\}^m}$ can, using the Holevo-Yuen-Kennedy-Lax optimality criterion \cite{Holevo.1973,Yuen.1975}, be seen to be exactly given by local Holevo-Helstrom measurements.
\end{remark}

\section{Conclusion and Outlook}
We have proposed a novel way of modifying the classical binary classification problem to obtain a quantum counterpart. The conceptual difference to the framework of quantum PAC learning as discussed in \cite{Arunachalam.2017} is that we work with maps whose outputs are themselves quantum states, not classical labels. This naturally gives rise to training data given by quantum states, which is one aspect in which our setting differs from \cite{Aaronson.2007}.\\

Using results from classical learning theory on dealing with classification noise in the training data, we exhibited learning strategies (based on the Holevo-Helstrom measurement) for binary classification with classical instances and quantum labels. The learning strategies consist of two main steps: First, classical information is extracted from the training data by performing a (localized) measurement. Second, classical learning strategies are applied. We complemented these procedures by sample complexity lower bounds thereby establishing the information-theoretic optimality of these strategies for pure label states w.r.t.~the dependence on VC-dimension, confidence and accuracy.\\

We conclude with some open questions that we leave open for further research:
\begin{itemize}
\item Can we derive sample complexity lower bounds which explicitly incorporate factors related to the hardness of distinguishing $\sigma_0$ and $\sigma_1$, e.g., in terms of $\norm{\sigma_0-\sigma_1}_1$ or $\max\lbrace \tr[E_0\sigma_1],\tr[E_1\sigma_0]\}$? Or can the corresponding factors in the upper bounds be eliminated? Could this be related to another complexity measure from classical learning theory, the ``fat-shattering dimension'' of the class $$\{\mathcal{X}\times\mathcal{E}(\IC^d) \ni (x,E)\mapsto \tr[Ef(x)]~|~f\in\mathcal{F}\}?$$
\item Our analysis is focused on the information-theoretic part of the learning problem, i.e., the sample complexity. Can we improve the computational complexity?
\item For deriving our sample complexity upper bounds, we used specific classical learning procedures applied to the post-measurement training data. In the agnostic case, we use empirical risk minimization, in the realizable case we use a combination of a minimum-disagreement approach with a subsampling procedure. In both cases, we decided for these algorithms to achieve the (essentially) optimal sample complexity characterized via the VC-dimension.\\
However, we could use other classical learning procedures for ``post-processing''. Can we identify situations in which procedures like structural risk minimization, compression schemes, or stable learning procedures yield useful sample complexity bounds?
\item We considered the case of classical instances. Can this be extended to a scenario of quantum instances with classical (or even quantum) labels? Whereas we were able to study the case of classical instances and quantum labels with methods from learning with label noise, once the instances themselves are quantum, we might have to employ ideas from learning models with restricted access to the instances such as that of ``learning with restricted focus of attention'' proposed in \cite{BenDavid.1998}.
\item Our strategy uses the Holevo-Helstrom measurement which can be understood as inducing the minimum amount of noise. However, in classical learning theory it is well known that adding noise to the training data can be helpful in preventing overfitting. In this spirit, can we justify other measurements than the Holevo-Helstrom measurement?
\item We assumed throughout our analysis that the learning algorithm has to output a hypothesis that maps into $\lbrace \sigma_0,\sigma_1\rbrace$. What if we allow for hypotheses that map into $\textrm{conv}\left(\lbrace \sigma_0,\sigma_1\rbrace\right)$ or $\mathcal{S}(\IC^d)$?
\item Finally, we assume throughout that the label states $\sigma_0$, $\sigma_1$ are known in advance. Can this assumption be removed? Here, it might be helpful that Theorem \ref{NoisySampleComplexityUpperBound} does not need explicit knowledge of the error rates $\eta_0$, $\eta_1$, but merely of an upper bound $\eta_b$ on them.
\end{itemize}

\vfill
\section*{Acknowledgements}
M.C.C.~wants to thank Michael M.~Wolf for suggesting this problem, Gael Sentís and Otfried Gühne for the opportunity to present and discuss the ideas of this paper at the University of Siegen, Srinivasan Arunachalam for his detailed feedback on an earlier draft, and Benedikt Graswald for discussions leading to Example \ref{ExmGroundStatePreparationLearningProblem}. Also, M.C.C.~thanks the anonymous reviewers at QTML $2020$ and at Springer Quantum Machine Intelligence for their suggestions.\\
Support from the TopMath Graduate Center of TUM the Graduate School at the Technische Universität München, Germany, from the TopMath Program at the Elite Network of Bavaria, and from the German Academic Scholarship Foundation (Studienstiftung des deutschen Volkes) is gratefully acknowledged.

\newpage
\setcounter{secnumdepth}{0}
\defbibheading{head}{\section{References}}
\printbibliography[heading=head]

\newpage
\appendix
\section*{Appendix}
\setcounter{secnumdepth}{2}
\section{Proofs}\label{SctProofs}
\textbf{Proof of Lemma \ref{LmmRademacherComplexityNoiseCorrectedClass}:}
Let $z=((x_i,y_i))_{i=1}^m\in (\mathcal{X}\times\{ 0,1\})^m$. If we use $\mathds{1}_{\tilde{f}(x_i)\neq y_i} = \frac{1 - (1-2\tilde{f}(x_i))(1-2y_i)}{2}$ and $\mathds{1}_{\tilde{f}(x_i)= y_i} = \frac{1 + (1-2\tilde{f}(x_i))(1-2y_i)}{2}$, then we can rewrite
\begin{align*}
\hat{\mathcal{R}}(\tilde{\mathcal{G}})
&= \mathbb{E}_{\sigma}[\sup\limits_{\tilde{f}\in\tilde{\mathcal{F}}} \frac{1}{m}\sum\limits_{i=1}^m \sigma_i \tilde{\ell}(\tilde{f}(x_i),y_i)]\\
&= \mathbb{E}_{\sigma}\left[\sup\limits_{\tilde{f}\in\tilde{\mathcal{F}}} \frac{1}{m}\sum\limits_{i=1}^m \sigma_i \frac{1}{1-\eta_0-\eta_1}\left( (1-\eta_{1\oplus y_i}) \frac{1 - (1-2\tilde{f}(x_i))(1-2y_i)}{2} - \eta_{y_i}\frac{1 + (1-2\tilde{f}(x_i))(1-2y_i)}{2} \right)\right].
\end{align*}
Next, we use that $\mathbb{E}_\sigma [\sigma_i]=0$ and that $\sigma_i$ and $(1-2y_i)\sigma_i$ have the same distribution for all $i$. With this we obtain from the above
\begin{align*}
\hat{\mathcal{R}}(\tilde{\mathcal{G}})
&= \frac{1}{1-\eta_0-\eta_1}\mathbb{E}_{\sigma}\left[\sup\limits_{\tilde{f}\in\tilde{\mathcal{F}}} \frac{1}{m}\sum\limits_{i=1}^m \sigma_i  (1-\eta_{1\oplus y_i} + \eta_{y_i})\tilde{f}(x_i)\right]\\
&= \frac{1}{2(1-\eta_0-\eta_1)}\mathbb{E}_{\sigma_2,\ldots,\sigma_m}\Big[\sup\limits_{\tilde{f},\tilde{f}'\in\tilde{\mathcal{F}}} \frac{1}{m}\underbrace{(1-\eta_{1\oplus y_1} + \eta_{y_1}) (\tilde{f}(x_1)-\tilde{f}'(x_1))}_{\leq 2\lvert \tilde{f}(x_1)-\tilde{f}'(x_1)\rvert}\\
&\hphantom{\frac{1}{2(1-\eta_0-\eta_1)}\mathbb{E}_{\sigma_2,\ldots,\sigma_m}\Big[\sup\limits_{\tilde{f},\tilde{f}'\in\tilde{\mathcal{F}}} }~+\frac{1}{m}\sum\limits_{i=2}^m \sigma_i  (1-\eta_{1\oplus y_i} + \eta_{y_i})(\tilde{f}(x_i)+\tilde{f}'(x_i))\Big]\\
&\leq \frac{1}{1-\eta_0-\eta_1}\mathbb{E}_{\sigma}\left[ \sup\limits_{\tilde{f}\in\tilde{\mathcal{F}}} \frac{2}{m}\sigma_1\tilde{f}(x_1) + \frac{1}{m}\sum\limits_{i=2}^m \sigma_i (1-\eta_{1\oplus y_i} + \eta_{y_i})\tilde{f}(x_i)\right],
\end{align*}
where the last step used that the expression is invariant w.r.t.~interchanging $\tilde{f}$ and $\tilde{f}'$, so we can drop the absolute value. Now we can iterate this reasoning for $i=2,\ldots,m$ and obtain
\begin{align*}
\hat{\mathcal{R}}(\tilde{\mathcal{G}}
\leq \frac{2}{1-\eta_0-\eta_1}\mathbb{E}_{\sigma}\left[ \sup\limits_{\tilde{f}\in\tilde{\mathcal{F}}}\frac{1}{m}\sum\limits_{i=1}^m \sigma_i \tilde{f}(x_i)\right]
= \frac{2}{1-\eta_0-\eta_1}\hat{\mathcal{R}}(\tilde{\mathcal{F}}),
\end{align*}
the desired inequality.\hfill $\blacksquare$\\

\textbf{Proof of Lemma \ref{LmmAgnosticConfidenceLowerBound}:}
As $\mathcal{F}$ is non-trivial, there exist concepts $f, g\in\mathcal{F}$ and a point $x\in\mathcal{X}$ s.t.~$f(x)=\sigma_0$ and $g(x)=\sigma_1$. Let $\lambda\in (0,1)$ (to be chosen appropriately later in the proof). Define probability distributions $\mu_{\pm}$ on $\mathcal{X}\times\mathcal{D}$ via
\begin{align*}
\mu_{\pm}(x,f(x)) = \frac{1\pm \lambda}{2},\quad \mu_{\pm}(x,g(x))=\frac{1\mp\lambda}{2}.
\end{align*}
The risk of a hypothesis $h\in\mathcal{D}^\mathcal{X}$ w.r.t.~these probability measures is given by
\begin{align*}
R_\pm (h) &= \frac{1\pm \lambda}{4}\norm{\sigma_0 - h(x)}_1 + \frac{1\mp\lambda}{4}\norm{\sigma_1 - h(x)}_1\\
&=\begin{cases}
\frac{1\pm \lambda}{4}\norm{\sigma_0 - \sigma_1}_1\quad \textrm{if } h(x) = \sigma_1\\
\frac{1\mp \lambda}{4}\norm{\sigma_0 - \sigma_1}_1\quad \textrm{if } h(x) = \sigma_0
\end{cases},
\end{align*}
in particular the optimal achievable risk is $\frac{1- \lambda}{4}\norm{\sigma_0 - \sigma_1}_1$. Note that a hypothesis which predicts the suboptimal label state for $x$ has an excess risk of
\begin{align*}
\frac{1+ \lambda}{4}\norm{\sigma_0 - \sigma_1} - \frac{1- \lambda}{4}\norm{\sigma_0 - \sigma_1}_1
= \frac{\lambda}{2}\norm{\sigma_0 - \sigma_1}_1.
\end{align*}
So if we pick $\lambda=\frac{\varepsilon}{2\norm{\sigma_0 - \sigma_1}_1}<1$, then in order to achieve an excess risk $\leq\varepsilon$ with probability $\geq 1-\delta$, the learning algorithm has to be able to distinguish between the underlying distributions $\mu_\pm$ with probability $\geq 1-\delta$.\\
As the algorithm has access to the underlying distribution only via the training data, this means that the algorithm has to be able to distinguish the corresponding training data ensembles with probability $\geq 1-\delta$. Here, we observe that the training data being drawn i.i.d. according to $\mu_\pm$ is equivalent to the learning algorithm having access to $m$ copies of the state
\begin{align*}
\rho_\pm := \mu_\pm (x,f(x)) |x\rangle\langle x|\otimes\sigma_0 + \mu_\pm (x,g(x)) |x\rangle\langle x|\otimes\sigma_1,
\end{align*}
because this mixed state simply describes the statistical mixture. The optimal success probability for distinguishing between two quantum states is a well-studied object in quantum information theory. It can be characterized by the trace distance between the two states and is given (in our case) by (see, e.g., \cite{Nielsen.2009})
\begin{align*}
p_{\textrm{opt}} = \frac{1}{2}(1+\frac{1}{2}\norm{\rho_+^{\otimes m} - \rho_-^{\otimes m}}_1).
\end{align*}
As the trace distance of tensor products is not that easy to deal with, we will instead work with the fidelity defined as
\begin{align*}
F(\rho,\sigma):= \tr[\sqrt{\rho^{\frac{1}{2}}\sigma\rho^{\frac{1}{2}}}].
\end{align*}
According to the Fuchs-van de Graaf inequalities we have
\begin{align*}
\frac{1}{2}\norm{\rho_+^{\otimes m} - \rho_-^{\otimes m}}_1 \leq \sqrt{1-F(\rho_+^{\otimes m}, \rho_-^{\otimes m})^2}
= \sqrt{1-F(\rho_+, \rho_-)^{2m}},
\end{align*}
where the last steps uses multiplicativity of the fidelity under tensor products. Now we require $p_\textrm{opt}\geq 1-\delta$ and rearrange to obtain
\begin{align*}
F(\rho_+, \rho_-)^{2m} \leq 4\delta(1-\delta)
\end{align*}
or equivalently after taking logarithms
\begin{align*}
m\geq\frac{\log(4\delta(1-\delta))}{\log(F(\rho_+, \rho_-)^2)}.
\end{align*}
By strong concavity of the fidelity, we have
\begin{align*}
F(\rho_+,\rho_-)
&\geq \sqrt{\frac{1+\lambda}{2}\frac{1-\lambda}{2}}F(|x\rangle\langle x| \otimes f(x), |x\rangle\langle x| \otimes f(x)) + \sqrt{\frac{1-\lambda}{2}\frac{1+\lambda}{2}}F(|x\rangle\langle x| \otimes g(x), |x\rangle\langle x| \otimes g(x))\\
&= \sqrt{1-\lambda^2}.
\end{align*}
This now implies
\begin{align*}
m&\geq\frac{\log(4\delta(1-\delta))}{\log(F(\rho_+, \rho_-)^2)} 
= \frac{\log\left( \frac{1}{4\delta (1-\delta)}\right)}{\log\left(\frac{1}{F(\rho_+,\rho_-)^2}\right)} 
\geq \frac{\log\left( \frac{1}{4\delta (1-\delta)}\right)}{\log\left(\frac{1}{1-\lambda^2}\right)}.
\end{align*}
Thus we obtain (after Taylor-expanding the logarithm in the denominator)
\begin{align*}
m\geq \Omega\left(\norm{\sigma_0-\sigma_1}_1^2\frac{\log\left(\frac{1}{\delta}\right)}{\varepsilon^2}\right),
\end{align*}
as desired.\hfill $\blacksquare$\\

\textbf{Proof of Lemma \ref{LmmEigenvaluesStatisticalMixture}:}
Pick an orthonormal basis $\lbrace |k\rangle\rbrace_{k=1,\ldots,n}$ of $\mathbb{C}^n$ s.t. $|\psi\rangle=|0\rangle$ and $|\phi\rangle=\cos(\varphi)|0\rangle + \sin(\varphi)|1\rangle$ for an angle $0\leq\varphi <2\pi$. Then, when restricting to the relevant subspace spanned by $|0\rangle$ and $|1\rangle$, we get
\begin{align*}
\rho|_{\textrm{span}\lbrace |0\rangle,|1\rangle\rbrace}
&=
\begin{pmatrix}
\alpha + \beta\cos^2(\varphi) & \beta\cos(\varphi)\sin(\varphi)\\
\beta\cos(\varphi)\sin(\varphi) & \beta\sin^2(\varphi)
\end{pmatrix}
=:A.
\end{align*}
We now easily see that
\begin{align*}
\det(A) = \alpha\beta\sin^2(\varphi) \overset{!}{=} \lambda_1\lambda_2\textrm{ and }
\tr[A] = \alpha + \beta \overset{!}{=} \lambda_1+\lambda_2,
\end{align*}
where $\lambda_1,\lambda_2$ are the two non-zero eigenvalues of $\rho$. We can solve the second of these two equations for $\lambda_2$ and plug this back into the first equation to obtain
\begin{align*}
\lambda_1^2 - \lambda_1 (\alpha+\beta) + \alpha\beta\sin^2(\varphi) = 0.
\end{align*}
We now solve this quadratic equation and obtain the two eigenvalues
\begin{align*}
\lambda_{1/2} 
= \frac{\alpha+\beta\pm\sqrt{\alpha^2 + \beta^2 + 2\alpha\beta(2\cos^2(\varphi)-1)}}{2}
= \frac{\alpha+\beta\pm\sqrt{(\alpha-\beta)^2 + 4\alpha\beta |\langle\psi |\phi\rangle|^2}}{2},
\end{align*}
where we used that $|\cos(\varphi)| = |\langle\psi |\phi\rangle|$.\hfill $\blacksquare$\\

\textbf{Detailed Proof of Theorem \ref{ThmAgnosticVCDimLowerBound}:}
Let $S=(s_1,\ldots,s_d)\in\mathcal{X}$ be a set shattered by $\tilde{\mathcal{F}}$, for each $a\in\lbrace 0,1\rbrace^d$ define the distribution $\mu_a$ on $\lbrace 1,\ldots,d\rbrace\times\lbrace 0,1\rbrace$ via
\begin{align*}
\mu_a(i,b) := \frac{1}{2d}\left( 1 + (-1)^{a_i + b} \frac{8\varepsilon}{\norm{\sigma_0-\sigma_1}_1}\right).
\end{align*}
Note that $\forall a\in\lbrace 0,1\rbrace^d\ \exists f_a\in\tilde{\mathcal{F}}: f_a(s_i)=a_i$ by shattering and that for each $a\in\lbrace 0,1\rbrace^d$, $f_a$ is a minimum-error concept w.r.t.~$\mu_a$ and a concept $f_{\tilde{a}}$ has additional error
\begin{align*}
d_H(a,\tilde{a})\frac{8\varepsilon}{d\norm{\sigma_0-\sigma_1}_1}\cdot\frac{\norm{\sigma_0-\sigma_1}_1}{2} 
= d_H(a,\tilde{a})\frac{4\varepsilon}{d}
\end{align*}
compared to $f_a$. Hence, in order to solve the learning problem with confidence $1-\delta$ and accuracy $\varepsilon$ the algorithm $\mathcal{A}$ has to output, with probability $\geq 1-\delta$, a hypothesis (generated from the training data arising from the underlying string) that when evaluated on $S$ yields a vector that is $\frac{d}{4}$-close to the underlying string in Hamming distance.\\

Let $A$ be a random variable distributed uniformly on $\lbrace 0,1\rbrace^d$ (corresponding to the unknown underlying string $a$). Let $B=B_1\ldots B_m$ be the training data with each example generated independently from $\mu_a$ described by the quantum ensemble
\begin{align*}
\mathcal{E}_a = \lbrace \mu_a (i,b), |s_i\rangle\langle s_i|\otimes \sigma_b\rbrace_{i=1,\ldots,d,\ b=0,1},
\end{align*}
or, equivalently, by the quantum state
\begin{align*}
\rho_a = \sum\limits_{i=1}^d |s_i\rangle\langle s_i|\otimes\left(\mu_a(i,0)\sigma_0 + \mu_a(i,1)\sigma_1\right).
\end{align*}
In particular, the composite system of underlying string and corresponding training data is described by the quantum state
\begin{align*}
\sigma_{AB} = \frac{1}{2^d}\sum\limits_{a\in\lbrace 0,1\rbrace^d} |a\rangle\langle a|\otimes \rho_a^{\otimes m}.
\end{align*} 

We follow the information-theoretic proof strategy from \cite{Arunachalam.2018}, i.e., we first show a lower bound on the mutual information $I(A:B)$ which arises from the learning requirement, then observe that $I(A:B)\leq m\cdot I(A:B_1)$ and finally upper-bound the mutual information $I(A:B_1)$.\\

First for the mutual information lower bound. Let $h(B)\in\lbrace 0,1\rbrace^d$ denote the label vector assigned to $S$ by the hypothesis produced by the learner upon input of training data $B$. Let $Z=\mathds{1}_{\lbrace R_{\mu_A} (h) - \inf\limits_{f\in\mathcal{F}}R_{\mu_A}(f)\leq\varepsilon\rbrace}$. If $Z=1$, then by the above deliberations we conclude $d_H(A,h(B))\leq\frac{d}{4}$ and thus, given $h(B)$, $A$ ranges over a set of size $\sum\limits_{i=0}^{\frac{d}{4}} \binom{n}{i}\leq 2^{H\left(\frac{1}{4}\right)d}$. Thus we get (using data processing and the definition of conditional entropy)
\begin{align*}
I(A:B) &\geq I(A:h(B)) = H(A)-H(A|h(B))\\
&\geq H(A)-H(A|h(B),Z) - H(Z)\\
&= H(A) - \underbrace{\mathbb{P}[Z=1]}_{\leq 1} \underbrace{H(A|h(B),Z=1)}_{\leq H\left(\tfrac{1}{4}\right)d} - \underbrace{\mathbb{P}[Z=0]}_{\leq\delta}\underbrace{H(A|h(B),Z=0)}_{\leq d} -\underbrace{H(Z)}_{\leq H(\delta)} \\
&\geq d - H\left(\tfrac{1}{4}\right)d - \delta d - H(\delta)\\
&= \left( 1 - H\left(\tfrac{1}{4}\right) - \delta\right) d - H(\delta),
\end{align*}
in particular $I(A:B)\geq \Omega(d)$. (Here we use our assumption on $\delta$.)\\

Now we show $I(A:B)\leq m\cdot I(A:B_1)$. We reproduce the reasoning provided in \cite{Arunachalam.2018} for completeness:
\begin{align*}
I(A:B) &= S(B) - S(B|A)\\
&= S(B) - \sum\limits_{i=1}^m S(B_i|A)\\\
&\leq \sum\limits_{i=1}^m S(B_i) - S(B_i|A)\\
&= \sum\limits_{i=1}^m I(A:B_1).
\end{align*}
Here, the first step is by definition, the second uses the product structure of the subsystem $B$, the third follows from subadditivity of the entropy and the last is again by definition.\\

And finally, we prove an upper bound on $I(A:B_1)$. To this end, we have to study the reduced state
\begin{align*}
\sigma_{AB_1} = \frac{1}{2^d}\sum\limits_{a\in\lbrace 0,1\rbrace^d} |a\rangle\langle a|\otimes \rho_a.
\end{align*}
More precisely, we have 
\begin{align*}
I(A:B_1) = S(A) + S(B_1) - S(AB_1),
\end{align*}
and thus have to study the entropies of $\sigma_{AB_1}$ as well as those of the reduced states $\sigma_A$ and $\sigma_{B_1}$. As $A\sim\textrm{Uniform}\left(\lbrace 0,1\rbrace^d\right)$, we have $S(A)=d$. Now we consider the reduced state
\begin{align*}
\sigma_{B_1} &= \frac{1}{2^d}\sum\limits_{a\in\lbrace 0,1\rbrace^d} \rho_a\\
&= \sum\limits_{i=1}^d |s_i\rangle\langle s_i|\otimes\left(\left(\frac{1}{2^d}\sum\limits_{a\in\lbrace 0,1\rbrace^d} \mu_a(i,0)\right)|\psi_0\rangle\langle\psi_0| + \left(\frac{1}{2^d}\sum\limits_{a\in\lbrace 0,1\rbrace^d}\mu_a(i,1)\right)|\psi_1\rangle\langle\psi_1|\right).
\end{align*}
Here, we have
\begin{align*}
\frac{1}{2^d}\sum\limits_{a\in\lbrace 0,1\rbrace^d} \mu_a(i,0) = \frac{1}{2d} = \frac{1}{2^d}\sum\limits_{a\in\lbrace 0,1\rbrace^d} \mu_a(i,1).
\end{align*}
By Lemma \ref{LmmEigenvaluesStatisticalMixture} we know that $\frac{1}{2d}|\psi_0\rangle\langle\psi_0| + \frac{1}{2d}|\psi_1\rangle\langle\psi_1|$ has non-zero eigenvalues $\mu_{1/2} = \frac{1}{2d}(1 \pm |\langle\psi_0|\psi_1\rangle|)$ and due to the block-diagonal structure of $\sigma_{B_1}$ we conclude that the non-zero eigenvalues of $\sigma_{B_1}$ are also $\mu_{1/2}$, each of multiplicity $d$. In particular, we have
\begin{align*}
S(\sigma_{B_1}) &= d\cdot(-\mu_1\log(\mu_1) - \lambda_2\log(\mu_2))\\
&= \log(2d) - \frac{1}{2}\left(\log(1-|\langle\psi_0|\psi_1\rangle|^2) + |\langle\psi_0|\psi_1\rangle|\log\left(\frac{1+|\langle\psi_0|\psi_1\rangle|}{1-|\langle\psi_0|\psi_1\rangle|}\right)\right).
\end{align*}
Similarly, we see that the non-zero eigenvalues of $\sigma_{AB_1}$ are 
\begin{align*}
\frac{1}{2^d}\lambda_{1/2} = \frac{1}{2^d}\cdot\frac{1}{2d}\left( 1 \pm |\langle\psi_0|\psi_1\rangle|\sqrt{1+\frac{64\varepsilon^2}{\norm{\sigma_0-\sigma_1}_1^2}\cdot\frac{1-|\langle\psi_0|\psi_1\rangle|^2}{|\langle\psi_0|\psi_1\rangle|^2}}\right),
\end{align*}
each of multiplicity $d\cdot 2^d$ and that therefore 
\begin{align*}
S(\sigma_{AB_1}) = d &+ \log(2d) - \frac{1}{2}\Bigg(\log\left(1-|\langle\psi_0|\psi_1\rangle|^2 \left(1+\frac{64\varepsilon^2}{\norm{\sigma_0-\sigma_1}_1^2}\cdot\frac{1-|\langle\psi_0|\psi_1\rangle|^2}{|\langle\psi_0|\psi_1\rangle|^2}\right)\right) \\
&+ |\langle\psi_0|\psi_1\rangle|\sqrt{1+\frac{64\varepsilon^2}{\norm{\sigma_0-\sigma_1}_1^2}\cdot\frac{1-|\langle\psi_0|\psi_1\rangle|^2}{|\langle\psi_0|\psi_1\rangle|^2}}\log\left(\frac{1+|\langle\psi_0|\psi_1\rangle|\sqrt{1+\frac{64\varepsilon^2}{\norm{\sigma_0-\sigma_1}_1^2}\cdot\frac{1-|\langle\psi_0|\psi_1\rangle|^2}{|\langle\psi_0|\psi_1\rangle|^2}}}{1-|\langle\psi_0|\psi_1\rangle|\sqrt{1+\frac{64\varepsilon^2}{\norm{\sigma_0-\sigma_1}_1^2}\cdot\frac{1-|\langle\psi_0|\psi_1\rangle|^2}{|\langle\psi_0|\psi_1\rangle|^2}}}\right)\Bigg).
\end{align*}
If we combine these expressions for the different entropies, we obtain
\begin{align*}
I(A:B_1) &= S(A) + S(B_1) - S(AB_1)\\
&= \frac{1}{2}\left(\log\left(1-|\langle\psi_0|\psi_1\rangle|^2 - \frac{64\varepsilon^2}{\norm{\sigma_0-\sigma_1}_1^2} (1-|\langle\psi_0|\psi_1\rangle|^2)\right) - \log\left(1-|\langle\psi_0|\psi_1\rangle|^2\right)\right)\\
&\hspace*{3mm} +\frac{|\langle\psi_0|\psi_1\rangle|}{2}\Bigg(\sqrt{1+\frac{64\varepsilon^2}{\norm{\sigma_0-\sigma_1}_1^2}\cdot\frac{1-|\langle\psi_0|\psi_1\rangle|^2}{|\langle\psi_0|\psi_1\rangle|^2}}\log\left(\frac{1+|\langle\psi_0|\psi_1\rangle|\sqrt{1+\frac{64\varepsilon^2}{\norm{\sigma_0-\sigma_1}_1^2}\cdot\frac{1-|\langle\psi_0|\psi_1\rangle|^2}{|\langle\psi_0|\psi_1\rangle|^2}}}{1-|\langle\psi_0|\psi_1\rangle|\sqrt{1+\frac{64\varepsilon^2}{\norm{\sigma_0-\sigma_1}_1^2}\cdot\frac{1-|\langle\psi_0|\psi_1\rangle|^2}{|\langle\psi_0|\psi_1\rangle|^2}}}\right) \\
&\hspace*{3mm} - \log\left(\frac{1+|\langle\psi_0|\psi_1\rangle|}{1-|\langle\psi_0|\psi_1\rangle|}\right)\Bigg).
\end{align*}
We now use Taylor's theorem to understand the scaling of the different terms with $\varepsilon$. First, we have (by Taylor-expanding $\log(1-|\langle\psi_0|\psi_1\rangle|^2 - x)$ around $x=0$)
\begin{align*}
\log &\left(1-|\langle\psi_0|\psi_1\rangle|^2 - \frac{64\varepsilon^2}{\norm{\sigma_0-\sigma_1}_1^2} (1-|\langle\psi_0|\psi_1\rangle|^2)\right) - \log\left(1-|\langle\psi_0|\psi_1\rangle|^2\right)\\
&= \frac{1}{1-|\langle\psi_0|\psi_1\rangle|^2}\cdot \frac{64\varepsilon^2}{\norm{\sigma_0-\sigma_1}_1^2} (1-|\langle\psi_0|\psi_1\rangle|^2) + \mathcal{O}(\varepsilon^4)\\
&= - \frac{64\varepsilon^2}{\norm{\sigma_0-\sigma_1}_1^2} + \mathcal{O}(\varepsilon^4).
\end{align*}
Moreover, using the Taylor expansions
\begin{align*}
\log\left(\frac{1+a\sqrt{1+x}}{1-a\sqrt{1+x}}\right) = \log\left(\frac{1+a}{1-a}\right) + \frac{ax}{1-a^2} + \mathcal{O}(x^2)
\end{align*}
around $x=0$ (with $a>0$) and
\begin{align*}
\sqrt{1+\frac{64\varepsilon^2}{\norm{\sigma_0-\sigma_1}_1^2}\cdot\frac{1-|\langle\psi_0|\psi_1\rangle|^2}{|\langle\psi_0|\psi_1\rangle|^2}}
&= 1 + \frac{1}{2}\cdot\frac{64\varepsilon^2}{\norm{\sigma_0-\sigma_1}_1^2}\cdot\frac{1-|\langle\psi_0|\psi_1\rangle|^2}{|\langle\psi_0|\psi_1\rangle|^2} + \mathcal{O}(\varepsilon^4)
\end{align*}
we now obtain
\begin{align*}
&\sqrt{1+\frac{64\varepsilon^2}{\norm{\sigma_0-\sigma_1}_1^2}\cdot\frac{1-|\langle\psi_0|\psi_1\rangle|^2}{|\langle\psi_0|\psi_1\rangle|^2}}\log\left(\frac{1+|\langle\psi_0|\psi_1\rangle|\sqrt{1+\frac{64\varepsilon^2}{\norm{\sigma_0-\sigma_1}_1^2}\cdot\frac{1-|\langle\psi_0|\psi_1\rangle|^2}{|\langle\psi_0|\psi_1\rangle|^2}}}{1-|\langle\psi_0|\psi_1\rangle|\sqrt{1+\frac{64\varepsilon^2}{\norm{\sigma_0-\sigma_1}_1^2}\cdot\frac{1-|\langle\psi_0|\psi_1\rangle|^2}{|\langle\psi_0|\psi_1\rangle|^2}}}\right) \\
&\hspace*{3mm} - \log\left(\frac{1+|\langle\psi_0|\psi_1\rangle|}{1-|\langle\psi_0|\psi_1\rangle|}\right)\\
= &\left(1 + \frac{1}{2}\cdot\frac{64\varepsilon^2}{\norm{\sigma_0-\sigma_1}_1^2}\cdot\frac{1-|\langle\psi_0|\psi_1\rangle|^2}{|\langle\psi_0|\psi_1\rangle|^2} + \mathcal{O}(\varepsilon^4)\right)\\
&\cdot\left(\log\left(\frac{1+|\langle\psi_0|\psi_1\rangle|}{1-|\langle\psi_0|\psi_1\rangle|}\right) + \frac{|\langle\psi_0|\psi_1\rangle|}{1-|\langle\psi_0|\psi_1\rangle|^2}\cdot \frac{64\varepsilon^2}{\norm{\sigma_0-\sigma_1}_1^2}\cdot\frac{1-|\langle\psi_0|\psi_1\rangle|^2}{|\langle\psi_0|\psi_1\rangle|^2} + \mathcal{O}(\varepsilon^4)\right)\\
&- \log\left(\frac{1+|\langle\psi_0|\psi_1\rangle|}{1-|\langle\psi_0|\psi_1\rangle|}\right)\\
= &\frac{64\varepsilon^2}{\norm{\sigma_0-\sigma_1}_1^2}\left(\frac{1}{|\langle\psi_0|\psi_1\rangle|} + \frac{1-|\langle\psi_0|\psi_1\rangle|^2}{2|\langle\psi_0|\psi_1\rangle|}\log\left(\frac{1+|\langle\psi_0|\psi_1\rangle|}{1-|\langle\psi_0|\psi_1\rangle|}\right)\right) + \mathcal{O}(\varepsilon^4).
\end{align*}
Plugging these approximations back in gives us
\begin{align*}
I(A:B_1) &= \frac{64\varepsilon^2}{\norm{\sigma_0-\sigma_1}_1^2}\cdot \frac{1-|\langle\psi_0|\psi_1\rangle|^2}{4|\langle\psi_0|\psi_1\rangle|}\log\left(\frac{1+|\langle\psi_0|\psi_1\rangle|}{1-|\langle\psi_0|\psi_1\rangle|}\right)+ \mathcal{O}(\varepsilon^4)
= \mathcal{O}(\varepsilon^2).
\end{align*}
Now combining our mutual information lower and upper bounds yields
\begin{align*}
\Omega(d) \leq I(A:B)\leq m \cdot I(A:B_1) \leq m \cdot\mathcal{O}(\varepsilon^2),
\end{align*}
which after rearranging becomes
\begin{align*}
m\geq \Omega\left(\frac{d}{\varepsilon^2}\right),
\end{align*}
as desired.\hfill $\blacksquare$\\

\textbf{Detailed Proof of Lemma \ref{LmmRealizableConfidenceLowerBound}:}
As $\mathcal{F}$ is non-trivial, there exist $f_1,f_2\in\mathcal{F}$ and $x_1,x_2\in\mathcal{X}$ s.t.~$f_1(x_1)=f_2(x_1)=\sigma_0$ and $f_1(x_2)=\sigma_0\neq\sigma_1=f_2(x_2)$. Now consider the distribution $\mu$ on $\mathcal{X}$ defined by 
\begin{align*}
\mu(x_1)=1-\lambda,\quad \mu(x_2)=\lambda,
\end{align*}
where $\lambda\in (0,1)$ is to be chosen later in the proof.\\
The risk of a hypothesis $h\in\mathcal{D}^\mathcal{X}$ w.r.t.~$\mu$ if the target concept is $f_i$ is given by
\begin{align*}
R_{\mu,f_i}(h) = \frac{1-\lambda}{2}\norm{h(x_1)-f_i(x_1)}_1 + \frac{\lambda}{2}\norm{h(x_2)-f_i(x_2)}_1,
\end{align*}
so in particular we have
\begin{align*}
R_{\mu,f_i}(f_j) = \begin{cases} 0 &\textrm{if } i=j \\ \frac{\lambda}{2}\norm{\sigma_0-\sigma_1}_1 &\textrm{if } i\neq j \end{cases}.
\end{align*}
So if we choose $\lambda=\frac{2\varepsilon}{\norm{\sigma_0-\sigma_1}_1}<1$, then the learning requirement for $\mathcal{A}$ implies that with probability $\geq 1-\delta$, $\mathcal{A}$ correctly identifies whether the target concept is $f_1$ or $f_2$.\\
As the algorithm has access to the underlying distribution only via the training data, this means that the algorithm has to be able to distinguish the corresponding training data ensembles with probability $\geq 1-\delta$. Here, we observe that the training data being drawn i.i.d.~according to $\mu_\pm$ is equivalent to the learning algorithm having access to $m$ copies of the state
\begin{align*}
\rho_i = (1-\lambda)|x_1\rangle\langle x_1|\otimes\sigma_0 + \lambda |x_2\rangle\langle x_2|\otimes f_i(x_2),\quad i=1,2.
\end{align*}
The optimal success probability for distinguishing between two quantum states is a well-studied object in quantum information theory. It can be characterized by the trace distance between the two states and is given (in our case) by (see \cite{Nielsen.2009})
\begin{align*}
p_{\textrm{opt}} = \frac{1}{2}(1+\frac{1}{2}\norm{\rho_1^{\otimes m} - \rho_2^{\otimes m}}_1).
\end{align*}
As the trace distance of tensor products is not that easy to deal with, we will instead work with the fidelity defined as $F(\rho,\sigma):= \tr[\sqrt{\rho^{\frac{1}{2}}\sigma\rho^{\frac{1}{2}}}].$ According to the Fuchs-van de Graaf inequalities (see Section $9.2.3$ in \cite{Nielsen.2009}) we have
\begin{align*}
\frac{1}{2}\norm{\rho_1^{\otimes m} - \rho_2^{\otimes m}}_1 \leq \sqrt{1-F(\rho_1^{\otimes m}, \rho_2^{\otimes m})^2}
= \sqrt{1-F(\rho_1, \rho_2)^{2m}},
\end{align*}
where the last steps uses multiplicativity of the fidelity under tensor products. Now we require $p_\textrm{opt}\geq 1-\delta$ and rearrange to obtain
\begin{align*}
F(\rho_1, \rho_2)^{2m} \leq 4\delta(1-\delta)
\end{align*}
or equivalently after taking logarithms
\begin{align*}
m\geq\frac{\log(4\delta(1-\delta))}{\log(F(\rho_1, \rho_2)^2)}.
\end{align*}
Now we use again the Fuchs-van de Graaf inequalities which tell us (after rearranging)
\begin{align*}
1- \frac{1}{2}\norm{\rho_1 - \rho_2}_1 \leq F(\rho_1, \rho_2)\leq \sqrt{1-\frac{1}{4}\norm{\rho_1 - \rho_2}_1^2}
\end{align*}
to obtain that
\begin{align*}
m&\geq\frac{\log(4\delta(1-\delta))}{\log(F(\rho_1, \rho_2)^2)} 
= \frac{\log\left( \frac{1}{4\delta (1-\delta)}\right)}{\log\left(\frac{1}{F(\rho_1,\rho_2)^2}\right)} 
\geq \frac{\log\left( \frac{1}{4\delta (1-\delta)}\right)}{\log\left(\frac{1}{(1- \frac{1}{2}\norm{\rho_1 - \rho_2}_1)^2} \right)}
\geq \frac{\log(4\delta(1-\delta))}{2\log(1- \frac{1}{2}\norm{\rho_1 - \rho_2}_1)}.
\end{align*}
It is easy to see that $\norm{\rho_1-\rho_2}_1 = \lambda \norm{\sigma_0 - \sigma_1}_1 = 2\varepsilon.$
Now Taylor expansion of the logarithm gives
\begin{align*}
m\geq \Omega\left(\frac{\log\left(\frac{1}{\delta}\right)}{\varepsilon}\right),
\end{align*}
as desired.\hfill $\blacksquare$\\

\textbf{Detailed Proof of Theorem \ref{ThmRealizableVCDimLowerBound}:}
Let $S=(s_0,\ldots,s_d)\in\mathcal{X}$ be a set shattered by $\tilde{\mathcal{F}}$, define
\begin{align*}
\mu(s_0) = 1-\lambda,\quad \mu(s_i)=\frac{\lambda}{d}\quad \forall 1\leq i\leq d,
\end{align*}
with $\lambda\in (0,1)$ to be chosen later. By shattering, $\forall a\in\lbrace 0,1\rbrace^d\ \exists f_a\in\tilde{\mathcal{F}}$ s.t.
\begin{align*}
f_a(s_0) = 0\quad\textrm{and}\quad f_a(s_i)=a_i\quad \forall 1\leq i\leq d.
\end{align*}
Observe that w.r.t.~a distribution $\mu$ and target concept $f_a$, another concept $f_b$ has error
\begin{align*}
d_H(a,b)\cdot \frac{\lambda}{d}\cdot\frac{\norm{\sigma_0-\sigma_1}_1}{2}.
\end{align*}
So if we pick $\lambda=\frac{8\varepsilon}{\norm{\sigma_0-\sigma_1}_1}$, then by the learning requirement, with probability $\geq 1-\delta$, $\mathcal{A}$ has to output a hypothesis $h$ that when evaluated on $S$ yields a label vector that is $\frac{d}{4}$-close to the true underlying string in Hamming distance.\\
Denote by $A\sim\textrm{Uniform}\left(\lbrace 0,1\rbrace^d\right)$ a random variable describing the unknown underlying string, let $B=B_1\ldots B_m$ be the corresponding quantum training data system. We want to repeat the three-step reasoning from the proof of Theorem \ref{ThmAgnosticVCDimLowerBound}. The first two steps work exactly as before. Step $3$ will be slightly different. Again we have
\begin{align*}
I(A:B_1)=S(A)+S(B_1)-S(AB_1),\quad\textrm{and}\quad S(A)=d.
\end{align*}
In this case, the relevant composite state is
\begin{align*}
\sigma_{AB_1} = \frac{1}{2^d}\sum\limits_{a\in\lbrace 0,1\rbrace^d} |a\rangle\langle a|\otimes \rho_a,
\end{align*}
where $\rho_a = \sum\limits_{j=0}^d \mu(s_j) |s_j\rangle\langle s_j|\otimes f_a(s_j) = (1-\lambda)|s_0\rangle\langle s_0|\otimes\sigma_0 + \frac{\lambda}{d}\sum\limits_{j=1}^d |s_j\rangle\langle s_j|\otimes\sigma_{a_j}$.\\
We now again use Lemma \ref{LmmEigenvaluesStatisticalMixture} to compute eigenvalues and thus entropies. (Here our assumption that $\sigma_0$ and $\sigma_1$ are pure enters the proof.) We obtain
\begin{itemize}
\item Each $\rho_a$ has non-zero eigenvalues $1-\lambda$ of multiplicity $1$ and $\frac{\lambda}{d}$ of multiplicity $d$.
\item $\sigma_{B_1}=\frac{1}{2^d}\sum\limits_{a\in\lbrace 0,1\rbrace^d}\left((1-\lambda)|s_0\rangle\langle s_0|\otimes\sigma_0 + \frac{\lambda}{d}\sum\limits_{j=1}^d |s_j\rangle\langle s_j|\otimes\sigma_{a_j}\right) = (1-\lambda)|s_0\rangle\langle s_0|\otimes\sigma_0 + \frac{\lambda}{d}\sum\limits_{j=1}^d |s_j\rangle\langle s_j|\otimes\left( \frac{1}{2}\sigma_0+\frac{1}{2}\sigma_1\right)$ has non-zero eigenvalues $1-\lambda$ of multiplicity $1$ and $\frac{\lambda}{d}\lambda_{1/2}$ of multiplicity $d$, where $\lambda_{1/2} = \frac{1\pm |\langle\psi_0 |\psi_1\rangle|}{2}$.
\item $\sigma_{AB_1}$ has non-zero eigenvalues $\frac{1}{2^d}(1-\lambda)$ of multiplicity $2^d$ and $\frac{\lambda}{d\cdot 2^d}$ of multiplicity $d\cdot 2^d$.
\end{itemize}
With this we can now compute the relevant entropies and obtain
\begin{align*}
S(B_1) &= S(\sigma_{B_1})
\\&= -(1-\lambda) \log(1-\lambda) + d\left( -\frac{\lambda}{d}\lambda_1\log\left(\frac{\lambda}{d}\lambda_1\right) - \frac{\lambda}{d}\lambda_2\log\left(\frac{\lambda}{d}\lambda_2\right)\right)\\
&= -(1-\lambda)\log(1-\lambda) - \lambda \left( \lambda_1 \log\left(\frac{\lambda}{d}\lambda_1\right) + \lambda_2\log\left(\frac{\lambda}{d}\lambda_2\right)\right),
\end{align*}
as well as
\begin{align*}
S(AB_1) &= S(\sigma_{AB_1})\\
&= 2^d \left( -\frac{1}{2^d}(1-\lambda)\log\left(\frac{1}{2^d}(1-\lambda)\right) - d\cdot\frac{\lambda}{d\cdot 2^d}\log\left(\frac{\lambda}{d\cdot 2^d}\right)\right)\\
&= -(1-\lambda)\log\left(\frac{1-\lambda}{2^d}\right) - \lambda\log\left(\frac{\lambda}{d\cdot 2^d}\right).
\end{align*}
Hence, we now have 
\begin{align*}
I(A:B_1)&=S(A)+S(B_1)-S(AB_1)\\
&= -\frac{\lambda}{2}\underbrace{\left( \log\left(\frac{1-|\langle\psi_0|\psi_1\rangle|^2}{4}\right) + |\langle\psi_0|\psi_1\rangle|\log\left(\frac{1+|\langle\psi_0|\psi_1\rangle|}{1-|\langle\psi_0|\psi_1\rangle|}\right)\right)}_{\leq 0 \textrm{ because } |\langle\psi_0|\psi_1\rangle|\in [0,1]}\\
&=\mathcal{O}(\varepsilon).
\end{align*}
Now we can finish the proof by combining steps $1$, $2$ and $3$ as before.\hfill $\blacksquare$

\section{A Physical Motivation for our Notion of Risk}\label{SctRiskMotivation}
In our definition of the risk $R_\mu$ we use the trace distance. As the latter is a well-established measure of distinguishability of quantum states, it presents itself as a natural candidate loss function. Here, we give a more explicit operational reasoning as to why we choose to use the trace distance.\\

Imagine the learning task as a competition between two parties, a learner and a teacher. We assume that both parties obey the laws of quantum physics. The teacher knows (a classical description of) the probability distribution $\mu\in\textrm{Prob}(\mathcal{X}\times\mathcal{D})$ and will provide corresponding training data to the learner during a training phase. The learner's goal is to persuade the teacher in a test phase that she has managed to learn the distribution $\mu$, which was unknown to her in advance, i.e., that she has produced a good hypothesis $h:\mathcal{X}\to\mathcal{D}$.\\
We first give an informal description of the test phase: The teacher prepares another (independent) example $(x,\rho)$ drawn from $\mu$. She then sends $x$ to the learner. The latter applies her hypothesis $h$ to prepare the quantum state $h(x)$ which she then sends back to the teacher. The teacher now uses this one copy of $h(x)$ and her knowledge of $\mu$ to evaluate whether the learner made a good prediction. As also the teacher is restricted by quantum theory, she can only do so by performing a measurement.\\

We now discuss the choice of measurement of the teacher in more detail. On the one hand, the teacher wants to maximize the probability of detecting a wrong prediction. On the other hand, she does not want to be unfair, so at the same time she tries to maximize the probability of detecting a correct prediction. In summary, the teacher wants to choose a $2$-outcome measurement $\{ E_{accept},E_{reject}\}$ that maximizes
\begin{align*}
\tr[E_{accept}\sigma_i] + \tr[E_{reject}\sigma_j],
\end{align*}
where $\sigma_i=\rho$ and $\sigma_j\in\mathcal{D}\setminus\{\rho\}$. As she knows (a classical description of) the state $\rho\in\mathcal{D}$ and that $h(x)\in\mathcal{D}$, she can achieve this by picking $\{ E_{accept},E_{reject}\}$ to be the optimal measurement for minimum error discrimination of $\mathcal{D}$ (where the states are taken with equal prior probabilities, see Theorem $3.4$ in \cite{Watrous.2018}). The measurement is basically the same independently of whether $\rho=\sigma_1$ or $\rho=\sigma_2$, only the outcome labels are interchanged.\\

Now the expected probability of the trainer rejecting the learner's prediction is
\begin{align*}
\int\limits_{X\times\mathcal{D}} \tr[E_{reject}(\rho) h(x)]~\mathrm{d}\mu(x,\rho).
\end{align*}
The optimal measurement satisfies 
\begin{align*}
\tr[E_{accept}\sigma_i] + \tr[E_{reject}\sigma_j]=\frac{1}{2}\left( 1 + \frac{1}{2}\norm{\sigma_0 - \sigma_1}_1\right).
\end{align*}
It is easy to see that under the additional assumption that $\sigma_0$ and $\sigma_1$ have the same purity, i.e., $\tr[\sigma_0^2]=\tr[\sigma_1^2]$, the rejection probabilities are symmetric, namely
\begin{align*}
\tr[E_{accept}\sigma_j] = \tr[E_{reject}\sigma_i] = \frac{1}{4}\left( 1 - \frac{1}{2}\norm{\sigma_0 - \sigma_1}_1\right)
\end{align*}
and similarly 
\begin{align*}
\tr[E_{accept}\sigma_i] = \tr[E_{reject}\sigma_j] = \frac{1}{4}\left( 1 + \frac{1}{2}\norm{\sigma_0 - \sigma_1}_1\right).
\end{align*}
With this we now obtain when comparing the achieved with the optimal expected rejection probability 
\begin{align*}
&\int\limits_{X\times\mathcal{D}} \tr[E_{reject}(\rho) h(x)]~\mathrm{d}\mu(x,\rho) - \inf\limits_{g:\mathcal{X}\to\mathcal{D}}\int\limits_{X\times\mathcal{D}} \tr[E_{reject}(\rho) g(x)]~\mathrm{d}\mu(x,\rho)\\
&= \int\limits_{X\times\mathcal{D}} \frac{\norm{\rho-h(x)}_1}{4}~\mathrm{d}\mu(x,\rho) \\
&= \frac{1}{2} R_\mu(h).
\end{align*}
So we have recovered our notion of risk, at least in the case of states of equal purity, from a more basic analysis of the test phase.\\

Note that a similar analysis could be performed also in the case of more than two quantum labels. There, the teacher's measurement would be the optimal measurement for minimum error discrimination of $\rho$ and $\frac{1}{\lvert \mathcal{D}\rvert - 1}\sum_{\sigma\in\mathcal{D}\setminus\{\rho\}} \sigma$. Unfortunately, no closed-form expressions for the corresponding success probabilities are known. We do, however, see that in this scenario, using the trace distance as loss function would be too pessimistic from the perspective of the learner. As the teacher does not know the prediction state prepared by the learner, the teacher has to solve a state discrimination problem taking into account all possible label states.

\section{The Holevo-Helstrom Strategy}\label{SctHelstrom}
The naive learning strategy based on the Holevo-Helstrom measurement is the following:

\begin{tcolorbox}[title={Holevo-Helstrom strategy}, breakable]
Given: Quantum training data $S=\{ (x_i,\rho_i)\}_{i=1}^m$\\
Output: Hypothesis $\hat{h}:\mathcal{X}\to\mathcal{D}$\\
Algorithm: 
\begin{enumerate}
\item For each $i$: Perform a Holevo-Helstrom measurement on $\rho_i$. Let $$y_i =\begin{cases} 1\ &\textrm{if } E_1~ \textrm{is accepted}\\ 0 &\textrm{if } E_1~ \textrm{is rejected}\end{cases}.$$
\item Let $\tilde{S}=\lbrace (x_i,y_i)\rbrace_{i=1}^m\in (\mathcal{X}\times\lbrace 0,1\rbrace)^m$. Then one can view $(x_i,y_i)$ as being drawn independently according to the probability measure $\nu$ on $\mathcal{X}\times\lbrace 0,1\rbrace$ which has 
\begin{align*}
&\nu_1(x) = \mu_1 (x) = \mu (x,\sigma_0) + \mu (x,\sigma_1)
\end{align*}
as the first marginal and 
\begin{align*}
\nu (y|x) = \ &\delta_{y0}\left(\mu (\sigma_1 |x) \tr[\sigma_1 E_0] + \mu (\sigma_0 |x) \tr[\sigma_0 E_0]\right)\\
		      +~ &\delta_{y1}\left(\mu (\sigma_1 |x) \tr[\sigma_1 E_1] + \mu (\sigma_0 |x) \tr[\sigma_0 E_1]\right).	
\end{align*}
as the conditional probability distribution of $y$ given $x$. 
\item Use a classical learning algorithm for binary classification to find $g\in\tilde{\mathcal{F}}:=\lbrace \tilde{f}:\mathcal{X}\to\lbrace 0,1\rbrace\ |\ \exists f\in\mathcal{F}:\ f(x)=\sigma_{\tilde{f}(x)}\ \forall x\in\mathcal{X}\rbrace$ s.t.~$\tilde{R}_\nu (g):= \mathbb{P}_{(x,y)\sim\nu} [y\neq g(x)]$ is small.
\item Define $h:\mathcal{X}\to\mathcal{D}$ via $h(x)=\sigma_{g(x)}$ and output $h$ as hypothesis.
\end{enumerate}
\end{tcolorbox}


The remainder of this section is devoted to studying the performance of this simple learning procedure. Note that we leave open for now the classical learning algorithm to be used, we first work towards characterizing the true risk $R_\mu (h)$ in terms of the intermediate classical risk $\tilde{R}_\nu (g)$.\\
In the following we will often make use of the fact that when identifying $i\leftrightarrow \sigma_i$, the probability measure $\mu$ on $\mathcal{X}\times\mathcal{D}$ gives rise to a probability measure on $\mathcal{X}\times\lbrace 0,1\rbrace$. We will abuse notation and also denote the latter measure by $\mu$, however, which measure is meant will always be clear from the context.\\

Recall that $R_\mu (h) = \frac{\lVert \sigma_0 - \sigma_1 \rVert}{2} \mathbb{P}_{(x,\rho)\sim\mu}[h(x)\neq \rho].$ We now derive a similar expression for $\tilde{R}_\nu (g)$.

\begin{lemma}\label{LmmHelstromRiskComparison}
With the notation as in the Holevo-Helstrom strategy (in particular $h(x)=\sigma_{g(x)}$) it holds that
\begin{align*}
\tilde{R}_\nu(g) = \frac{\lVert\sigma_0-\sigma_1\rVert_1}{2}\mathbb{P}_{(x,\rho)\sim\mu}[h(x)\neq\rho] + \tr[\sigma_0 E_1] + (\tr[\sigma_1 E_0]-\tr[\sigma_0 E_1])\mathbb{E}_{\mu_1}[g].
\end{align*}
\end{lemma}
\begin{proof}
This can be shown by direct computation using the definition of $\nu$:
\begin{alignat*}{3}
& \tilde{R}_\nu(g) &&= \int\limits_{\mathcal{X}\times\lbrace 0,1\rbrace}&&\lvert y-g(x)\rvert\mathrm{d}\nu(x,y)\\
& &&=  \int\limits_{\mathcal{X}}&&\Big( \int\limits_{\lbrace 0,1\rbrace}\lvert y-g(x)\rvert\mathrm{d}\nu(y|x)\Big)\mathrm{d}\nu_1(x)\\
& &&= \int\limits_{\mathcal{X}}\Big(&&\lvert 1-g(x)\rvert(\mu(\sigma_1|x)\tr[\sigma_1 E_1]+\mu(\sigma_0|x)\tr[\sigma_0 E_1]) + \\
& && &&\lvert g(x)\rvert(\mu(\sigma_1|x)\tr[\sigma_1 E_0]+\mu(\sigma_0|x)\tr[\sigma_0 E_0])\Big)\mathrm{d}\mu_1(x)
\end{alignat*}
Now we use the specific property of the Holevo-Helstrom measurement that $\tr[(\sigma_1-\sigma_0)E_1]=\frac{\lVert\sigma_0-\sigma_1\rVert}{2}$. Moreover, as $g(x)\in\lbrace 0,1\rbrace$, we have $\lvert 1-g(x)\rvert= 1-g(x)$ and $\lvert g(x)\rvert = g(x)$. Thus we obtain
\begin{alignat*}{3}
& \tilde{R}_\nu(g) &&= \frac{\lVert\sigma_0-\sigma_1\rVert}{2}&&\int\limits_{\mathcal{X}}\Big((1-g(x))\mu(\sigma_1|x)+g(x)\mu(\sigma_0|x)\Big)\mathrm{d}\mu_1(x) + \\
& && &&\int\limits_{\mathcal{X}}\Big((1-g(x))\tr[\sigma_0 E_1] + g(x)\tr[\sigma_1 E_0]\Big)\mathrm{d}\mu_1(x)\\
& &&= \frac{\lVert\sigma_0-\sigma_1\rVert}{2}&&\mathbb{P}_{(x,\rho)\sim\mu}[h(x)\neq\rho]+ \tr[\sigma_0 E_1]  +(\tr[\sigma_1 E_0]-\tr[\sigma_0 E_1])\mathbb{E}_{\mu_1}[g],
\end{alignat*}
where the last step uses $h(x)=\sigma_{g(x)}$.
\end{proof}
\\

This allows us to easily compare the true and the intermediate risk and obtain
\begin{align*}
\tilde{R}_\nu (g)-R_\mu(h)
&= \tr[\sigma_0 E_1](1-2\mathbb{E}_{\mu_1}[g]) + \Big(1-\frac{\lVert\sigma_0-\sigma_1\rVert}{2}\Big)\mathbb{E}_{\mu_1}[g].
\end{align*}
As $g(x)\in\lbrace 0,1\rbrace\ \forall x\in\mathcal{X}$ and in particular $0\leq\mathbb{E}_{\mu_1}[g]\leq 1$, this gives rise to the following

\begin{corollary}
With the notation as in the Holevo-Helstrom strategy it holds that
\begin{align*}
\tilde{R}_\nu (g) - \max\lbrace\tr[\sigma_0 E_1],\tr[\sigma_1 E_0]\rbrace \leq R_\mu(h)\leq \tilde{R}_\nu (g) - \min\lbrace\tr[\sigma_0 E_1],\tr[\sigma_1 E_0]\rbrace.
\end{align*}
\end{corollary}

We can extend this to a comparison between the excess risks
\begin{align*}
R_\mu (h) - R^*_{\mu,\mathcal{F}}:=R_\mu (h) - \inf\limits_{\eta\in\mathcal{F}} R_\mu(\eta)\ \textrm{and}\ \tilde{R}_\nu (g) - \tilde{R}^*_{\nu,\tilde{\mathcal{F}}} := \tilde{R}_\nu (g) - \inf\limits_{\gamma\in\tilde{\mathcal{F}}}\tilde{R}_\nu (\gamma)
\end{align*}
which are the quantities of interest for agnostic learning scenarios.

\begin{corollary}
With the notation as in the Holevo-Helstrom strategy it holds that
\begin{align*}
\tilde{R}_\nu (g) - \tilde{R}^*_{\nu,\tilde{\mathcal{F}}} - \lvert\tr[\sigma_0 E_1]-\tr[\sigma_1 E_0]\rvert &\leq R_\mu (h)-R^*_{\mu,\mathcal{F}}\\
&\leq \tilde{R}_\nu (g) - \tilde{R}^*_{\nu,\tilde{\mathcal{F}}} + \lvert\tr[\sigma_0 E_1]-\tr[\sigma_1 E_0]\rvert
\end{align*}
\end{corollary}

So we see that solving the classical learning task in step $3$ of the Holevo-Helstrom strategy does not necessarily imply success at the overall learning task if the target accuracy is $\varepsilon<\lvert\tr[\sigma_0 E_1]-\tr[\sigma_1 E_0]\rvert$. This problem is addressed by the noise-corrected Holevo-Helstrom strategy presented in Section \ref{SctUpperBounds}.

\begin{remark}
We want to shortly discuss a special case in which the connection between $R_\mu (h)$ and $\tilde{R}_\nu (g)$ takes a particularly appealing form. Namely, assume that $\sigma_0$ and $\sigma_1$ are such that the corresponding Holevo-Helstrom measurement produces equal probabilities of error, i.e., $\tr[E_0\sigma_1]=\tr[E_1\sigma_0]$. This is clearly not true in general, take, e.g., $\sigma_0=|0\rangle\langle 0|$ and $\sigma_1=\frac{1}{2}(|0\rangle\langle 0|+|1\rangle\langle 1|)$. It does, however, hold true in certain special cases, e.g., if both $\sigma_0$ and $\sigma_1$ are pure or if $\sigma_0$ and $\sigma_1$ have the same (non-trivial) purity and $\tr[E_0]=\tr[E_1]$. (The latter is, e.g., satisfied if $\sigma_0$ and $\sigma_1$ are qubit states of the same (non-zero) purity.)\\
In this simple case our previous discussion yields $R_\mu(h)=\tilde{R}_\nu (g)$, in particular, if we succeed at the classical binary classification task in step $3$, then we also succeed at the overall classification task with quantum labels, so the quantum learning task is reduced to a classical learning problem.
\end{remark}

\section{Sample Complexity of Binary Classification with Two-Sided Classification Noise}\label{SctNoisyBinaryClassification}
Here, we discuss the sample complexity of the PAC learning task of binary classification in the presence of (two-sided) classification noise in the realizable scenario. To be in congruence with the literature on this and related problems, we will use a slightly different notation than in the main body of the paper. Namely, we will consider classical input space $\mathcal{X}$ and classical target space $\lbrace 0,1\rbrace$, a concept class $\mathcal{F}\subset\lbrace 0,1\rbrace^\mathcal{X}$, a probability measure $\mu\in \textrm{Prob}(\mathcal{X})$, and noise probabilities $0\leq\eta_0,\eta_1<\frac{1}{2}$, with which labels are flipped. Moreover, we will work with the $0$-$1$-loss function and denote the corresponding risk of a hypothesis $h$ w.r.t.~a target concept $f$ by $\textrm{err}_\mu(h; f)=\mu [h(x)\neq f(x)]$. Finally, any training data sample $S$ splits the concept class $\mathcal{F}$ into so-called $S$-equivalence classes, where $f_1,f_2\in\mathcal{F}$ are equivalent if and only if $f_1(x)=f_2(x)\ \forall x\in\mathcal{X}$ s.t. $\exists y\in\lbrace 0,1\rbrace$ with $(x,y)\in S$.\\
The basic learning strategy underlying our discussion is Algorithm \ref{MinDisagreementAlgo}. It is the natural analog of searching for a consistent function in the case of noisy labels. Namely, as such a consistent function will in general not exist, it searches for a function that disagrees with the training data on as few examples as possible.

\begin{algorithm}
\caption{Minimum Disagreement Strategy $L$ (Algorithm $5.6$ in \cite{Laird.1988})}\label{MinDisagreementAlgo}
\begin{flushleft}
\hspace*{\algorithmicindent} \textbf{Input}: confidence and accuracy parameters $0<\delta,\varepsilon\leq\frac{1}{2}$, a noise threshold $0\leq\eta_0,\eta_1\leq\eta_b<\frac{1}{2}$ and noisy training data $S=\lbrace (x_i,y_i)\rbrace_{i=1}^m$ created from $\mu\in\textrm{Prob}(\mathcal{X})$ and some $f\in\mathcal{F}$ where $0$-labels are corrupted with prob. $\eta_0$ and $1$-labels are corrupted with prob. $\eta_1$, where
\begin{align*}
m\geq &\underbrace{\max\lbrace\frac{8}{\varepsilon}\log\left(\frac{6}{\delta}\right), \frac{16d}{\varepsilon}\log\left(\frac{16d}{\varepsilon}\right)\rbrace}_{=:m_1}\\
+ &\underbrace{\frac{2}{\varepsilon (1-\exp(-\frac{1}{2}(1-2\eta_b)^2))}\ln\left(\frac{1}{d}\left(\max\lbrace\frac{8}{\varepsilon}\log\left(\frac{6}{\delta}\right), \frac{16d}{\varepsilon}\log\left(\frac{16d}{\varepsilon}\right)\rbrace^d + 1\right)\right)}_{=:m_2}.
\end{align*}
\hspace*{\algorithmicindent} \textbf{Output}: a hypothesis $h\in\mathcal{F}$.
\end{flushleft}
\begin{algorithmic}[1]
\State Let $S_1$ consist of the first $m_1$ examples in $S$. Let $S_2=S\setminus S_1$.
\State Select $\mathcal{F}_1=\lbrace f_1,\ldots,f_N\rbrace$ as representatives of the $S_1$-equivalence classes induced by $S_1$, where $N\leq (m_1)^d+1$.
\State Output a hypothesis in $\mathcal{F}_1$ which minimizes the number of disagreements with $S_2$.
\end{algorithmic}
\end{algorithm}


\begin{theorem}\emph{(see Theorems $5.7$ and $5.33$ in \cite{Laird.1988})}\\
The output hypothesis $h$ of Algorithm \ref{MinDisagreementAlgo} satisfies $\textrm{err}_\mu(h; f) \leq\varepsilon.$
\end{theorem}

Laird's original proof that this algorithm solves the PAC learning problem is for the case $\eta_0=\eta_1$. It is, however, easily generalized to our case because we still assume the same noise bound on both error rates. (We only have to adapt the expression for the error rate and the corresponding Hoeffding bounds.)\\

In order to apply the reasoning by \cite{Hanneke.2016} we need to slightly reformulate the result of this algorithm s.t.~we obtain a bound on the error in terms of the sample size. When following the proof of Theorem $5.7$ in \cite{Laird.1988} we see that $m_1$ is used to ensure that there is a hypothesis which performs better than some given error threshold and $m_2$ is used to ensure that such a hypothesis is actually chosen. In particular, if we use the error bound by \cite{Blumer.1989} in terms of the sample size, we see that $m_2$ depends on $m_1$ as follows:
\begin{align*}
m_2
&= \frac{2}{1-\exp(-\frac{1}{2}(1-2\eta_b)^2)}\cdot \frac{m_1}{2}\cdot \frac{1}{d\log\left(\frac{2em_1}{d}\right)+\log\left(\frac{2}{\delta}\right)} \cdot \ln\left(\frac{1}{\delta} (m_1^d + 1)\right).
\end{align*}

\begin{remark}
Note that we cannot directly use the tighter error bound in terms of the sample complexity proved by \cite{Hanneke.2016} here because Laird's proof explicitly makes use of the strategy employed by \cite{Blumer.1989} which works via consistency with a given training sample.
\end{remark}

We can now easily bound
\begin{align*}
m= m_1 + m_2
\leq m_1 \cdot\left(1 + \frac{1}{1-\exp(-\frac{1}{2}(1-2\eta_b)^2)} \cdot\frac{1}{\log(e)}\cdot\frac{1}{1 - \frac{d\log\left(\frac{d}{2e}\right)}{d \log\left( m_1\right) + \log\left(\frac{2}{\delta}\right)}}\right).
\end{align*}

If we now further assume that $\delta>0$ is chosen s.t. $\log\left(\frac{2}{\delta}\right)>2d\log\left(\frac{d}{2e}\right)$, then we can continue upper-bounding this and obtain
\begin{align*}
m = m_1 + m_2 \leq (1 + C(\eta_b))m_1,
\end{align*}
where we defined $C(\eta_b):=\frac{2}{1-\exp(-\frac{1}{2}(1-2\eta_b)^2)}$. It is easy to check that for $0\leq \eta_b < \frac{1}{2}$, $C(\eta_b)\leq \frac{4}{(1-2\eta_b)^2}$, which well be used later on.\\
Hence, using a sample of size $m\geq 2(1+C(\eta_b))$ for the minimum disagreement strategy with $m_2=\lceil\frac{C(\eta_b)}{1 + C(\eta_b)} m\rceil$ and $m_1=m-m_2$ gives - using $\frac{m}{2(1+C(\eta_b))}\leq m_1\leq \frac{m}{1+C(\eta_b)}\leq \frac{m_2}{C(\eta_b)}$ - an error guarantee of
\begin{align}
\textrm{err}_\mu (h;f^*)
&\leq \frac{4}{m_1}\left(d\log\left(\frac{2em_1}{d}\right)+\log\left(\frac{2}{\delta}\right)\right)\\
&\leq \frac{8\cdot (1 + C(\eta_b))}{m}\left(d\log\left(\frac{2em}{d\cdot (1+C(\eta_b))}\right)+\log\left(\frac{2}{\delta}\right)\right).\label{ErrBoundBaseLearner}
\end{align}

With this suboptimal base learner we will now follow the strategy by \cite{Hanneke.2016} in order to build a better learner from it. Note that Hanneke's proof includes several steps in which the existence of a function consistent with the respective subsample is ensured. This is not necessary in our case because the Minimum Disagreement Strategy does not require a consistent function to exist.\\

We recall the algorithm for preprocessing the training data to generate subsamples as introduced in \cite{Hanneke.2016} in our Algorithm \ref{SubsampleAlgo}.

\begin{algorithm}
\caption{Subsample Generation Algorithm $\mathbb{A}(\cdot,\cdot)$ \cite{Hanneke.2016} \label{SubsampleAlgo}}
\begin{flushleft}
\hspace*{\algorithmicindent} \textbf{Input}: two finite sets $S$ and $T$.\\
\hspace*{\algorithmicindent} \textbf{Output}: a finite set $\mathbb{A}(S;T)$ of subsets of $S\cup T$.
\end{flushleft}
\begin{algorithmic}[1]
\If{$|S|\leq 3$,}
	\State Output $\lbrace S\cup T\rbrace$.
\Else
	\State Divide $S=\lbrace s_1,\ldots,s_{|S|}\rbrace$ into subsets in the following way: 
    \begin{align*}
    S_0&=\lbrace s_1,\ldots,s_{|S|-3\lfloor |S|/4 \rfloor}\rbrace,\\
    S_1&=\lbrace s_{|S|-3\lfloor |S|/4\rfloor + 1},\ldots,s_{|S|-2\lfloor |S|/4\rfloor}\rbrace,\\ 
    S_2&=\lbrace s_{|S|-2\lfloor |S|/4\rfloor + 1},\ldots,s_{|S|-\lfloor |S|/4\rfloor}\rbrace,\\
    S_3&=\lbrace s_{|S|-\lfloor |S|/4\rfloor + 1},\ldots,s_{|S|}\rbrace.
    \end{align*}
\EndIf
\State Return $\mathbb{A}(S_0;S_2\cup S_3\cup T) \cup \mathbb{A}(S_0;S_1\cup S_3\cup T)\cup \mathbb{A}(S_0;S_1\cup S_2\cup T)$.
\end{algorithmic}
\end{algorithm}


\begin{theorem}\label{NoisySampleComplexityUpperBound}
Let $\varepsilon\in (0,1)$, $\delta\in (0,2\cdot (\tfrac{2e}{d})^d)$ and $\eta_b \in (0,\tfrac{1}{2})$. Let $\mathcal{F}\subset \lbrace 0,1\rbrace^\mathcal{X}$ be a function class of VC-dimension $d$. Then $m=m(\varepsilon,\delta)=\mathcal{O}\left(\frac{1}{\varepsilon (1-2\eta_b)^2} \left( d + \log\left(\frac{1}{\delta}\right)\right)\right)$ noisy examples from a function in $\mathcal{F}$ are sufficient for binary classification in the presence of two-sided classification noise with error probabilities $0\leq\eta_0,\eta_1 <\eta_b$ with accuracy $\varepsilon$ and confidence $1-\delta$.\\
\end{theorem}
\begin{proof}
This proof is analogous to the proof of Theorem $2$ in \cite{Hanneke.2016} with some minor simplifications and adaptations and is given here only for the sake of completeness.\\
Fix an $f^*\in\mathcal{F}$ and a probability measure $\mu$ over $\mathcal{X}$. Denote by $S=S_{1:m}$ the corresponding noisy training data. For any classifier $h$ denote by $ER(h)=\lbrace x\in\mathcal{X}|h(x)\neq f^*(x)\rbrace$ the set of instances on which $h$ errs.\\
Fix $c=7200$. We will show by strong induction that $\forall m'\in\IN$, $\forall\delta'\in(0,\ldots)$ and for all finite sequences $T'$ with probability $\geq 1-\delta'$ the classifier
\begin{align*}
\hat{h}_{m',T'} = \textrm{Majority}\big( L(\mathbb{A}(S_{1:m'};T'))\big)
\end{align*}
satisfies the error bound
\begin{align}\label{ErrBoundImprovedClassifier}
\textrm{err}_\mu(\hat{h}_{m',T'}, f^*)\leq \frac{cC(\eta_b)}{1+m'}\left( d + \ln\left(\frac{18}{\delta'}\right)\right).
\end{align}
As base case consider $m'\leq C(\eta_b)c\cdot\ln(18e) - 1$. In this case, for any $\delta'\in(0,1)$ and for any finite sequence $T'$, we trivially have
\begin{align*}
\textrm{err}_\mu(\hat{h}_{m',T'}, f^*)
&\leq 1\\
&\leq \frac{c\cdot C(\eta_b)}{1+m'}\left( d + \ln(18)\right)\\
&\leq \frac{c\cdot C(\eta_b)}{1+m'}\left( d + \ln\left(\frac{18}{\delta'}\right)\right),
\end{align*}
as desired.\\

For the induction step, assume that for some $m>C(\eta_b)c\cdot\ln(18e) - 1$ for all $m'\in\IN$ with $m'<m$, for all $\delta'(0,2\cdot (\tfrac{2e}{d})^d)$ and for all finite sequences $T'$ with probability $\geq 1-\delta'$, (\ref{ErrBoundImprovedClassifier}) holds.\\
Note that by our choice of $c$ we have $C(\eta_b)c\cdot\ln(18e) - 1\geq 3$. Thus $|S_{1:m}|\geq 4$ and therefore $\mathbb{A}\left(S_{1:m};T\right)$ returns in step $3$. Let $S_0,S_1,S_2,S_3$ be as in $\mathbb{A}(S;T)$. Denote $T_1=S_2\cup S_3\cup T$, $T_2=S_1\cup S_3\cup T$, $T_3=S_1\cup S_2\cup T$ and $h_i = \textrm{Majority}\left( L (\mathbb{A}(S_0;T_i))\right)$ for each $i\in\lbrace 1,2,3\rbrace$.\\
Note that $S_0=S_{1:(m-3\lfloor \frac{m}{4}\rfloor)}$. As $m\geq 4$, $1\leq m-3\lfloor \frac{m}{4}\rfloor < m$. Also, $h_i = \hat{h}_{(m-3\lfloor \frac{m}{4}\rfloor),T_i}$. So by the induction hypothesis applied under the conditional distribution given $S_1,S_2,S_3$, which are independent of $S_0$, combined with the law of total probability, for every $i\in\lbrace 1,2,3\rbrace$ there exists an event $E_i$ of probability $\geq 1-\frac{\delta}{9}$ on which
\begin{align}\label{SingleErrBoundInductionStep}
\mu[ER(h_i)]
\leq \frac{cC(\eta_b)}{1 + |S_0|} \left( d + \ln\left(\frac{9\cdot 18}{\delta}\right)\right)
\leq \frac{4cC(\eta_b)}{m}\left( d + \ln\left(\frac{9\cdot 18}{\delta}\right)\right).
\end{align}
Next, fix an $i\in\lbrace 1,2,3\rbrace$ and write $\lbrace (\tilde{X}_{i,1},\tilde{Y}_{i,1}),\ldots,(\tilde{X}_{i_,N_i},\tilde{Y}_{i,N_i})\rbrace := S_i\cap (ER(h_i)\times\mathcal{Y})$. As $h_i$ and $S_i$ are independent, $\tilde{X}_{i,1},\ldots,\tilde{X}_{i,N_i}$ are conditionally independent given $h_i$ and $N_i$. Therefore we can apply the error bound \eqref{ErrBoundBaseLearner} for our base learner $L$ under the conditional distribution given $h_i$ and $N_i$ to conclude: There exists an event $E'_i$ of probability $\geq 1-\frac{\delta}{9}$ s.t., if $N_i>0$, then the output $h$ of the base learner $L$ upon input of $S_i\cap (ER(h_i)\times\mathcal{Y})$ satisfies
\begin{align*}
\textrm{err}_{\mu(\cdot|ER(h_i))}(h,f^*)
\leq \frac{8 ( 1+ C(\eta_b))}{N_i}\left(d\log\left(\frac{2eN_i}{d( 1+ C(\eta_b))}\right) + \log\left(\frac{18}{\delta}\right)\right).
\end{align*}
In particular, on $E'_i$ (if $N_i>0$) every $h\in\bigcup\limits_{j\in\lbrace 1,2,3\rbrace\setminus\lbrace i\rbrace} L\left(\mathbb{A}(S_0;T_j)\right)$ satisfies
\begin{align}
\mu[ER(h)\cap ER(h_i)]
&= \mu[ER(h_i)]\mu[ER(h)|ER(h_i)]\\
&= \mu[ER(h_i)] \textrm{err}_{\mu(\cdot | ER(h_i))}(h,f^*)\\
&\leq \mu[ER(h_i)]\frac{8 ( 1+ C(\eta_b))}{N_i}\left(d\log\left(\frac{2eN_i}{d( 1+ C(\eta_b))}\right) + \log\left(\frac{18}{\delta}\right)\right). \label{JointErrBoundInductionStep}
\end{align}
Using Chernoff bounds we get that there exists an event $E''_i$ of probability $\geq 1-\frac{\delta}{9}$ s.t., if $\mu[ER(h_i)]\geq \frac{2 (\tfrac{10}{3})^2}{\lfloor \frac{m}{4}\rfloor} \ln\left(\frac{9}{\delta}\right)$, then $N_i\geq\frac{7}{10}\mu[ER(h_i)]\lfloor \frac{m}{4}\rfloor$. In particular, on $E''_i$ we have the implication
\begin{align*}
\mu[ER(h_i)]\geq \frac{2 (\tfrac{10}{3})^2}{\lfloor \frac{m}{4}\rfloor} \ln\left(\frac{9}{\delta}\right)\ 
\Rightarrow \
N_i>0.
\end{align*}
If we now combine this with (\ref{SingleErrBoundInductionStep}) and (\ref{JointErrBoundInductionStep}), then we see: On $E_i\cap E'_i \cap E''_i$, if $\mu[ER(h_i)]\geq \frac{2 (\tfrac{10}{3})^2}{\lfloor \frac{m}{4}\rfloor} \ln\left(\frac{9}{\delta}\right)$, then every $h\in\bigcup\limits_{j\in\lbrace 1,2,3\rbrace\setminus\lbrace i\rbrace} L\left(\mathbb{A}(S_0;T_j)\right)$ satisfies
\begin{align*}
\mu[ER(h)\cap ER(h_i)]
&\leq \frac{80\cdot C(\eta_b)}{7 \lfloor \frac{m}{4}\rfloor} \left( d\log\left(\frac{2e\cdot\tfrac{7}{10}\cdot \mu[ER(h_i)]\lfloor \frac{m}{4}\rfloor}{dC(\eta_b)}\right) + \log\left(\frac{18}{\delta}\right)\right)\\
&\leq \frac{80\cdot C(\eta_b)}{7 \lfloor \frac{m}{4}\rfloor} \left( d\log\left(\frac{\tfrac{7e}{5}\cdot c\left(d+\ln\left(\tfrac{9\cdot 18}{\delta}\right)\right)}{d}\right) + \log\left(\frac{18}{\delta}\right)\right)\\
&\leq \frac{80\cdot C(\eta_b)}{7 \lfloor \frac{m}{4}\rfloor} \left( d\log\left(\frac{2}{5}c\left(\frac{7}{2}e + \frac{7e}{d}\ln\left(\frac{18}{\delta}\right)\right)\right) + \log\left(\frac{18}{\delta}\right)\right)\\
&\leq \frac{80\cdot C(\eta_b)}{7\ln(2) \lfloor \frac{m}{4}\rfloor} \left( d\ln\left(\frac{9ec}{5}\right) + 8 \ln\left(\frac{18}{\delta}\right)\right),
\end{align*}
where the last step uses the technical Lemma $5$ from the Appendix of \cite{Hanneke.2016}. As $m>C(\eta_b)c\cdot\ln(18e) - 1>3200$, we have $\lfloor\frac{m}{4}\rfloor>\frac{m-4}{4}>\frac{799}{800}\frac{m}{4} >\frac{799}{800}\frac{3200}{3201}\frac{m+1}{4}$. We use this relaxation and compute the logarithmic factors to obtain from the above that
\begin{align*}
\mu[ER(h)\cap ER(h_i)]
\leq \frac{600 \cdot C(\eta_b)}{m+1}\left( d + \ln\left(\frac{18}{\delta}\right)\right).
\end{align*}
Moreover, if $\mu[ER(h_i)]<\frac{23}{\lfloor \frac{m}{4}\rfloor}\ln\left(\frac{9}{\delta}\right)$, then simply because $\mu$ is a probability measure, we conclude
\begin{align*}
\mu[ER(h)\cap ER(h_i)] 
\leq \mu[ER(h_i)]
< \frac{23}{\lfloor \frac{m}{4}\rfloor}\ln\left(\frac{9}{\delta}\right)
< \frac{600\cdot C(\eta_b)}{m+1}\left( d + \ln\left(\frac{18}{\delta}\right)\right).
\end{align*}
Hence, no matter what value $\mu[ER(h_i)]$ takes, on the event $E_i\cap E'_i\cap E''_i$ we have for all $h\in\bigcup\limits_{j\in\lbrace 1,2,3\rbrace\setminus\lbrace i\rbrace} L\left(\mathbb{A}(S_0;T_j)\right)$ that 
\begin{align*}
\mu[ER(h)\cap ER(h_i)]
\leq \frac{600\cdot C(\eta_b)}{m+1}\left( d + \ln\left(\frac{18}{\delta}\right)\right).
\end{align*}
Now denote $h_{\textrm{maj}} = \hat{h}_{m,T} = \textrm{Majority}(L(\mathbb{A}(S;T)))$ for $S=S_{1:m}$. By definition of the majority function, for any $x\in\mathcal{X}$ at least $\frac{1}{2}$ of the classifiers $h$ in the sequence $L(\mathbb{A}(S;T))$ satisfy $h(x)=h_{\textrm{maj}}(x)$. So by the strong form of the pigeon hole principle, there exists an $i\in\lbrace 1,2,3\rbrace$ s.t. $h_i(x)=h_{\textrm{maj}}(x)$. Also, since each $\mathbb{A}(S_0;T_j)$ contributes an equal number of entries to $\mathbb{A}(S;T)$, for each $i\in\lbrace 1,2,3\rbrace$, at least $\frac{1}{4}$ of the classifiers $h\in\bigcup\limits_{j\in\lbrace 1,2,3\rbrace\setminus\lbrace i\rbrace} L\left(\mathbb{A}(S_0;T_j)\right)$ satisfy $h(x)=h_{\textrm{maj}}(x)$.\\
In particular, if $I$ is a random variable independent of the training data and distributed uniformly on $\lbrace 1,2,3\rbrace$ and if $\tilde{h}$ is a random variable conditionally given $I$ and $S$ uniformly distributed on $\bigcup\limits_{j\in\lbrace 1,2,3\rbrace\setminus\lbrace I\rbrace} L\left(\mathbb{A}(S_0;T_j)\right)$, then for any fixed $x\in ER(h_{\textrm{maj}})$, with conditional probability $\geq\frac{1}{12}$, $h_I(x)=\tilde{h}(x)=h_{\textrm{maj}}(x)$ and thus $x\in ER(h_I)\cap ER(\tilde{h})$.\\
Hence, for a random variable $X\sim\mu$ independent of the data, of $I$ and of $\tilde{h}$ we can now conclude
\begin{align*}
\mathbb{E}[\mu[ER(h_i)]\cap ER(\tilde{h}))|S]
&= \mathbb{E}[\mathbb{P}[X\in ER(h_I)\cap ER(\tilde{h})|I,\tilde{h},S]|S]\\
&= \mathbb{E}[\mathds{1}_{ X\in ER(h_I)\cap ER(\tilde{h})}|S]\\
&= \mathbb{E}[\mathbb{P}[X\in ER(h_I)\cap ER(\tilde{h})|S,X]|S]\\
&\geq \mathbb{E}[\mathbb{P}[X\in ER(h_I)\cap ER(\tilde{h})|S,X]\mathds{1}_{X\in ER(h_{\textrm{maj}})}|S]\\
&\geq \mathbb{E}[\tfrac{1}{12}\mathds{1}_{X\in ER(h_{\textrm{maj}})}|S]\\
&\geq \frac{1}{12}\textrm{err}_\mu (h_{\textrm{maj}};f^*).
\end{align*}
So on the event $\bigcap\limits_{i\in\lbrace 1,2,3\rbrace} E_i\cap E'_i \cap E''_i$ it holds that
\begin{align*}
\textrm{err}_\mu (h_{\textrm{maj}};f^*)
&\leq 12 \mathbb{E}[\mu[ER(h_i)\cap ER(\tilde{h})]|S]\\
&\leq 12 \max\limits_{i\in\lbrace 1,2,3\rbrace} \max\limits_{j\in\lbrace 1,2,3\rbrace\setminus\lbrace i\rbrace} \max\limits_{h\in L(\mathbb{A}(S_0;T_j))} \mu[ER(h_i)\cap ER(h)]\\
&< \frac{7200 \cdot C(\eta_b)}{m+1}\left(d + \ln\left(\frac{18}{\delta}\right)\right)\\
&= \frac{c\cdot C(\eta_b)}{m+1}\left(d + \ln\left(\frac{18}{\delta}\right)\right).
\end{align*}
Since by the union bound the event $\bigcap\limits_{i\in\lbrace 1,2,3\rbrace} E_i\cap E'_i \cap E''_i$ has probability $\geq 1-\delta$, the induction step is complete.\\

It remains to use the claim just proven by induction to derive the desired sample complexity upper bound. For this, take $T=\emptyset$ and note that for $m\geq\lfloor \frac{cC(\eta)}{\varepsilon}\left(d + \ln\left(\frac{18}{\delta}\right)\right)\rfloor$ the right hand side of (\ref{ErrBoundImprovedClassifier}) is $\leq\varepsilon$. Therefore such a sample size suffices for successful learning using $\textrm{Majority}(L(\mathbb{A}(\cdot;\emptyset)))$. Now recall the discussion before the Theorem, where we observed that $C(\eta_b)\leq \frac{4}{(1-2\eta_b)^2}$, to finish the proof.
\end{proof}

\end{document}